\documentclass[12pt]{article}
\usepackage{amsmath,amssymb,amsthm,booktabs}
\usepackage{mathrsfs}
\usepackage[OT1]{fontenc}
\usepackage[utf8]{inputenc}
\usepackage[colorlinks,citecolor=blue,urlcolor=blue]{hyperref}
\usepackage{txfonts}
\usepackage{indentfirst}
\usepackage{multirow}
\usepackage{float,subfig}

\usepackage{bm}
\usepackage{euscript}
\usepackage{graphicx}
\usepackage{multicol}
\usepackage[usenames,dvipsnames,svgnames,table]{xcolor}

\usepackage[round]{natbib}
\numberwithin{equation}{section} \theoremstyle{plain}

\newtheorem{proposition}{Proposition}[section]

\newtheorem{assumption}{Assumption}[section]
\newtheorem{example}{Example}[section]
\newtheorem{definition}{Definition}[section]
\newtheorem{remark}{Remark}[section]
\newtheorem{condition}{Condition}[section]

\newcommand{\gai}[1]{{#1}}

\makeatletter
\def\ps@pprintTitle{%
  \let\@oddhead\@empty
  \let\@evenhead\@empty
  \let\@oddfoot\@empty
  \let\@evenfoot\@oddfoot
}
\makeatother

\begin{document}

\newcommand\tabfig[1]{\vskip5mm \centerline{\textsc{Insert #1 around here}}  \vskip5mm}

\vskip2cm

\begin{center}
  \begin{minipage}{0.9\linewidth}
    \begin{flushleft}
      {\Large\bf    Improving Value-at-Risk prediction  under model uncertainty} \\[5mm]
      {\sc Shige Peng}\\
      \quad Institute of Mathematics
      \\ \quad Shandong University
      \\[2mm]
      {\sc Shuzhen Yang}\\
      \quad ZhongTai Securities Institute for Financial Studies
      \\
      \quad Shandong University
      \\[2mm]
      {\sc Jianfeng Yao}\\
      \quad Department of Statistics and Actuarial Science
      \\
      \quad The University of Hong Kong
      \\[2cm]
      {\em Correspondence:} \\
      \quad {\sc Jianfeng Yao}\\
      \quad Department of Statistics and Actuarial Science \\
      \quad The University of Hong Kong \\
      \quad Pokfulam Road\\
      \quad HONG KONG SAR \\
      \quad Email: \texttt{\small jeffyao@hku.hk}
    \end{flushleft}
  \end{minipage}
\end{center}

\newpage

\begin{abstract}
  Several well-established benchmark predictors exist for
   Value-at-Risk (VaR), a major  instrument for financial
  risk management. Hybrid methods combining
  AR-GARCH filtering with skewed-$t$ residuals and the extreme value
  theory-based approach are particularly recommended.
  This study introduces yet another VaR predictor,   G-VaR, which follows
  a novel methodology. Inspired by the recent mathematical theory of
  sublinear expectation, G-VaR is built
  upon the concept of model uncertainty, which in the present case
  signifies that  the inherent volatility of financial returns cannot be
  characterized  by a single distribution but rather by infinitely many statistical
  distributions.
  By considering the worst scenario among these potential distributions,
  the G-VaR predictor is precisely identified.
  Extensive experiments
  on both the NASDAQ Composite Index and  S\&P500 Index demonstrate
  the  excellent performance of the G-VaR predictor, which is
  superior to most existing  benchmark  VaR predictors.
\end{abstract}

\bigskip


\bigskip



\section{Introduction}

Since its birth at J.P. Morgan in the 1990s,
 value-at-risk (VaR)  has  become one of
the most used  (if not THE most used) instruments for assessing
downside risk in  financial markets.  Every unit of risk management in
today's financial industry routinely implements several  VaR indicators to monitor its business \citep{Jorion07}.
The regulatory authorities also incorporate VaR measures into their recommendations to
the banking industry (Basel Accords I-III), which has accelerated the
spread of  VaR.

The success of VaR methodology is also {backed} up by   a  rich body of
literature in which the methodology is carefully evaluated and discussed in different
model settings and for different markets and products.
The  literature on VaR is  voluminous and includes several
specialized books. Also any textbook treating  financial econometrics or risk
{management will have a}  chapter dedicated to  VaR.
For an up-to-date  account of this literature, the reader is referred to the recent review papers by
\cite{Kuester06},
\cite{Jorion10},
\cite{Abad14},
\cite{NadaChan16}, and
\cite{Zhang17}, among others.
In particular, Table 4 in \cite{Abad14}
lists as many as fourteen papers  that survey  and
compare different VaR methodologies  through empirical studies.
Further,  by focusing on  univariate observations, \cite{Kuester06}
offer a rich  review of  mainstream VaR measures and
provide an extensive empirical comparison of those measures in terms of their
prediction power using the daily NASDAQ Composite Index.
The general conclusion they draw   \citep[see also][]{Abad14} is that
whatever  method  is used for VaR modeling, the predictions are always improved,
 most of the time considerably
improved,  by applying that method  to residuals {\em filtered} by an
AR-GARCH model instead of  the original series $(r_t)$.
For example, one of
the best performers is obtained by applying extreme value
theory (EVT) to the
residuals of  AR-GARCH fit using skewed-$t$ innovations (AR-GARCH
St-EVT).
\cite{Kuester06} conclude that, at least for the NASDAQ Composite Index,
``conditionally heteroskedastic models yield acceptable forecasts'' and
that the conditional skewed-$t$ (AR-GARCH-St) together with
the conditional skewed-$t$ coupled with EVT  (AR-GARCH-St-EVT) perform
best in general.

In this paper, we present an entirely  new type of VaR. Our
methodology is inspired by  a rigorous mathematical theory called
{\em nonlinear expectation}.
{The theory of  {nonlinear expectation} was originally introduced in}
\cite{Peng2004,Peng2006,Peng2008,Peng2010}.
The  part of the theory relevant to
VaR prediction is detailed  in the appendix.
When applied to  analysis of a time-series of returns $\{r_t\}$,
the central
concept of {nonlinear expectation theory}  is {\em an {infinite} family of distributions} inherent in the
return series $\{r_t\}$.
 Traditional econometric modeling commonly assumes
that returns,  at least during certain time periods,
obey   {\em one} stochastic process  governed by {\em one}
stochastic-process model  $P_0$.  The
task of the econometrician  is to infer this unknown, but {\em true},
$P_0$.   The distribution can be parametric, as in an AR-ARCH model (with
skewed-$t$ or normal innovations),  or made up by a family of
conditional mixture distributions (see Section~\ref{sec:review}).  It can also be
fully nonparametric without any particular model specification,  as in the  historical simulation (HS) approach  to VaR
prediction.  However, a unique  stochastic model  is assumed {so far} for the
returns $\{r_t\}$.
The point of view  of {nonlinear expectation} theory is radically
different:  instead of assuming the existence of one unique  model
$P_0$,  it views   returns as originating from  a large number of
different models, say $\{P_\theta\}_{\theta\in\Theta}$, and this family
of potential models is indeed infinite (here, $\Theta$
denotes some  imprecise index set).
The rationale is that data  under investigation such as return series
are of a complex nature such that no single stochastic model or
distribution can
serve as a perfect model:
{\em model uncertainty} has to be considered, and any statistical
inference has to take into account such uncertainty.
We name this vision of complex data and the implied methodology
{\em data  analysis under model uncertainty}.

The concept of model uncertainty, sometimes also referred to as {\em model
ambiguity}, has taken a long time to emerge. An early attempt in this
direction was made in the area of robust statistics, where it was argued
that a statistical procedure (e.g. parameter estimation or hypothesis
testing)  can gain robustness by assuming that the data follow not a
single distribution but rather a family of distributions
$\{F_\theta\}_{\theta\in\Theta} $
\citep[see][]{Huber81,Walley91}.
{
  {By allowing a variable range  between two extreme values
    $\mathrm{\sigma}_{min}$ and $\mathrm{\sigma}_{max}$ for volatility
    of a stock},
  \cite{ALP95} presented a model for pricing and hedging derivative
  securities and option portfolios. When the volatility is unknown and
  takes value in some convex region { that can vary with the
    price process and  time}, \cite{L95} introduced optimal and risk-free strategies for intermediaries in such markets to meet their obligations.}
\cite{P97} later proposed a formal mathematical  approach to  model uncertainty, with
a nonlinear expectation called $g$-expectation
 introduced to develop the concept of
{\em  mean uncertainty} and its associated mathematical tools.
The  $g$-expectation concept was  then adopted by
\cite{CE02}
to describe the continuous-time inter-temporal version of
multiple-priors utility.  In particular, they  established
a separate premium for ambiguity on top of the traditional  premium
for risk.
In addition, \cite{EJ13} formulated a
model of utility in  a continuous-time framework that captures
aversion to ambiguity about both the volatility and the mean of
returns via {the theory of {\em sublinear expectation}
  (SLE). Note that sublinear expectation is a useful form of  nonlinear
  expectation reinforced by the addition of a sub-addivity property.}
All of the above
theories formalize an  inherent  family of distributions involving a   set of probability measures
$\{P_{\theta}\}_{\theta\in \Theta}$,  not just one probability
measure $P_0$, that governs the
statistical distributions of a  {dataset}.
In  related work, \cite{A99} proposed the concept of {\em coherent
  risk measures}, { which focused on measurement of both market risks and nonmarket risks,
 see also  \cite{FS11}. }

Although the vision of model uncertainty has been formalized through the rigorous
mathematical theory of SLE, its implications for real
data analysis  have not been fully explored. { Early work by Huber
\citep{Huber81} studied the model uncertainty in the area
of robust \gai{statistics}. \cite{C06} considered a quantitative framework
for defining model uncertainty in option pricing models, and
illustrated the difference between model uncertainty and the more
common notion of "market risk" through examples. By distinguishing
estimation risk and misspecification risk, \cite{KMS10} proposed a
procedure to take model risk into account in the computation of
capital reserves which can be used to address Value-at-Risk and
expected shortfall.} To the best of our knowledge, \gai{the present paper}
on VaR prediction constitutes the first
attempt at real-life data analysis under SLE-based model uncertainty.
The implementation of  SLE theory herein leads to a new
type of VaR predictor called  {\em $G$-VaR}.
Loosely speaking, as long as VaR prediction is concerned and to give
model uncertainty in the form of an infinite family of probability models
$\{P_\theta\}_{\theta\in\Theta} $,  G-VaR concentrates on
prediction under  the {\em worst scenario} among  all potential models $\{P_\theta\}$.
Extensive empirical analyses of two major  market indexes, namely, the
NASDAQ Composite  Index and S\&P
500 Index,  establish the superiority of  the new
G-VaR predictions over several benchmark  VaR predictors
that are among the best performers reported in \cite{Kuester06}.
The uniform superiority of  G-VaR {in these empirical studies} is truly astonishing. One {\em a
posteriori}  explanation is that these return data do have the kind
of complex nature that can be better understood through the lens of
model uncertainty that had led to the  G-VaR predictor.


The remainder of the paper is organized as follows.
Section~\ref{sec:review} briefly reviews several benchmark predictors
for the VaR of return series.
Section~\ref{sec:Gnormal} introduces the concepts of distribution
family as model uncertainty and  G-normal distribution.
{Section~\ref{sec:GVaR}  contains the main technical contribution
  of the paper.   The new VaR predictor, i.e. G-VaR, is
  introduced under  model uncertainty and its theoretical properties established.
}
The  implementation of G-VaR
is presented in Section~\ref{sec:implem}, in
which  consistent estimators are proposed for the parameters
involved in the G-VaR  predictor.
{Note that this implementation is far non trivial and
  Section~\ref{sec:implem} contains the main methodological
  contribution of the paper.
}
Section~\ref{sec:empirical} reports the  empirical results of the
G-VaR predictor for   the
NASDAQ Composite Index and S\&P 500 Index, with
extensive comparison made with the benchmark VaR predictors reviewed
in Section~\ref{sec:review}, including
the AR(1)-GARCH(1,1)-Normal, AR(1)-GARCH(1,1)-Skewed-t and
AR(1)-GARCH(1,1)-Skewed-t-EVT predictors. Finally,
Section~\ref{sec:concl} concludes.

\section{A brief review of benchmark predictors for  VaR}
\label{sec:review}

Before introducing our new methodology, we first
give a brief review of several of the  well-documented  VaR predictors described in
\cite{Kuester06}, as they  will serve as  benchmarks for comparison with
the new VaR predictor proposed herein.
As historical
references for these VaR measures  can be found in the earlier review paper,
  only a few key references are indicated in this  brief review.
Here, $(r_t)$ denotes a univariate time-series  for which VaR
prediction is required. In most common situations, the series
represents the daily returns of a market, or price of a stock or product.

\medskip\noindent{ \em (i) The Historical Simulation method.}\quad  This traditional method uses  sample
  quantiles from historical data to predict  VaR.  The method  is well documented  in classic books such as
  \cite{Dowd02} and \cite{Christoffersen03}.

  A variant of  HS is  filtered historical simulation (FHS), whereby the
  sample quantiles are calculated from filtered residuals using a
  parametric model such as the  AR-GARCH model.  Classical references on
  FHS  include \cite{BGV99,BGV02}.

\medskip\noindent{\em (ii) The method of Peaks over Thresholds using EVT.}\quad
  EVT provides a method for estimating the high upper
  quantiles of a variable $X$, say quantiles $x_\alpha$ such that
  $\overline{F}_X(x)=P(X>x_\alpha)=\alpha$ for some small tail
  probability $\alpha < 0.1$. From sample data, one
  obtains  an empirical quantile $u$, {\em the threshold}, such that
  $P(X>u)\approx 0.1$.
  Also those values above  threshold $u$ provide a
  sample for the ``survival distribution'' $ (X | X>u) $. EVT
  ensures  that for a large enough $u$, the survival distribution can
  be approximated by a generalized Pareto distribution (GPD) that
  depends on a pair of shape-scale parameters $(\beta,\xi)$
  \citep{Pickands75,EKM97}.
  This GPD is thus identified using the sample, which leads to an
  estimate for the initial tail probability $\overline{F}_X(x_\alpha)$ for all $x_\alpha>u$.

  For VaR prediction, the above procedure is applied to the available data
  $\{X_s= -r_s, ~s<t\}$
  to find the corresponding upper quantile $x_\alpha$. The negative
  operation is used here, as EVT considers  upper quantiles whereas VaR
  targets lower quantiles. Then, the
  prediction for  VaR is    $\widehat{\text{VaR}}_{\alpha}(t) = -
  x_{\alpha} $ for all risk levels $\alpha <0.1$. Empirical studies on VaR
  predictions using EVT can be found in \cite{McNeil00} and \cite{Kuester06}.

\medskip\noindent{\em (iii) The method by AR-GARCH filtering (AR-GARCH).}\quad
Here the  returns are assumed to follow a mean-variance decomposition of
type
  \begin{equation}
    \label{eq:mean-var}
    r_t = \mu_t + \sigma_t z_t,
  \end{equation}
  where the mean process $(\mu_t)$ follows an AR model (or, more
  generally, an ARMA model) and the residual is modeled by a GARCH
  process \citep{Bollerslev86} with independent and identically  distributed (i.i.d.)
  innovations $(z_t)$. Common choices for the distribution
  $f_z$ of the innovations $(z_t)$ are (i) standard normals,
  (ii) Student's $t$-distributions, and (iii) skewed
  $t$-distributions.

  Once the model \eqref{eq:mean-var} is fitted, a parametric estimate
  of the  distribution $f_z$ is obtained, say $\hat f_z$, which leads
  to an estimated quantile function, say $\hat Q_\alpha(z)$ (for any
  given risk level $\alpha$).  A level-$\alpha$ VaR prediction at time $t$ is
  thus defined as
  \[     \widehat{\text{VaR}}_{\alpha} (t)=
    - \left\{\hat{\mu}_{t}  +\hat \sigma_{t}\hat Q_\alpha(z) \right\}.
  \]

\medskip\noindent{\em (iv) The method of Conditional Mixture Modeling.}\quad
In this approach, conditional to the
  information set  ${\cal F}_{t-1}$ at time $t$, the return  $r_t$
  follows a mixture distribution with $n$ components, each of which has
  a constant mean parameter $\mu_{j}$ ($1\le j\le n$) and  time-varying
  volatility (variance) parameter $\sigma^2_{j,t}$ ($1\le j\le
  n$).
  Moreover, the $n$-dimensional volatility process
  $\boldsymbol{\sigma}_t^2=(\sigma^2_{1,t},\ldots,\sigma^2_{n,t})$ obeys a
  multidimensional GARCH equation.
  Here, the distributions of the mixture components are usually taken to be normal
  distributions or generalized exponential distributions (GED). More
  references to this approach can be found in \cite{Haas04}.

\medskip\noindent{\em (v) The method by Quantile Regression.}\quad
Here a regression model is used to predict the
VaR at time $t$ using some predictable covariates ${\bf X}_t\in{\cal
    F}_{t-1}$; see
  \cite{Koenker78} and \cite{Chernozhukov01}.
  Later,  \cite{Engle04} proposed  CAViaR models, where the quantiles (or VaRs) follow an
  autoregressive model without any exogeneous covariate.

\section{G-normal distribution}
\label{sec:Gnormal}

When measuring the risk in a financial time-series $\{X_t\}_{0\leq t\leq
  T}$, it is commonly assumed that the data follow a certain  distribution
$F_0$, and the aim is   to  estimate or approximate the ``true'' distribution $F_0$ or
some of its  characteristics such as  VaR. As we saw in the survey in
Section~\ref{sec:review}, many different VaR measures exist in the literature,
including  HS, EVT-VaR and their AR-GARCH-filtered variants.
Now we view the data $\{X_t\}_{0\leq t\leq T}$ as possessing a complex
nature governed  not by one distribution $F_0$ but rather by an infinite family
of distributions $\{F_{\theta}\}_{\theta\in \Theta}$, each of them
capturing some  properties of
the data.

How can the VaR concept  be  extended to a new framework  in which  an
infinite family of (unknown) distributions governs the data?
This paper provides an answer to this general
question.
To proceed, some  meaningful characteristics of the  family
$\{F_{\theta}\}_{\theta\in \Theta}$ need to be identified.  Consider
the simplest features of the distributions, namely,  their mean
$\mu$ and  variance $\sigma^2$. In general, these features
are time-varying.
Here, we consider a simple case, where  mean
$\mu$ is constant (independent of time) and  variance $\sigma^2$ is
time-varying within some
interval
$[\underline{\sigma}^2,\overline{\sigma}^2]$. In the following,
the interval $[\underline{\sigma}^2,\overline{\sigma}^2]$ is used to
characterize the unknown family of distributions $\{F_{\theta}\}_{\theta\in \Theta}$. For a given
canonical  probability space $(\Omega,\mathcal{F},P)$,
$\Omega=C([0,T])$, and Brownian motion $\{B_t\}_{0\leq t\leq T}$,  we
define the probability measures $P_{\theta}$ as follows.  For $A\in\mathcal{F}$,
\[
P_{\theta}( A )=P \circ Y_{\theta}^{-1}(A)=P(Y_{\theta}\in A),
\]
where
$$
Y_{\theta}(\cdot)=\int_0^{\centerdot}\theta_sdB_s,\quad \theta\in
\Theta=L^2_{\gai{\cal F}}(\Omega\times[0,T],[\underline{\sigma},\overline{\sigma}]),
$$
\gai{
where $\Theta$ is the set of all progressively measurable processes
taking value on $[\underline{\sigma},\overline{\sigma}]$}
 The collection of  $P_{\theta}$s is denoted as  $\{P_{\theta}\}_{\theta\in \Theta}$.
 Let the mean $\mu$ be $0$ and the distribution of $\{X_t\}_{0\leq
   t\leq T}$ under $P_{\theta}$ be  $F_{\theta}$. Thus, this
 \gai{infinite} family of
 stochastic process  distributions $\{F_{\theta}\}_{\theta\in \Theta}$
 is chosen as the family governing the  dataset
 $\{X_t\}_{0\leq t\leq T}$. In this paper,  we use the so-called
 G-normal distribution
 $N(0,[\underline{\sigma}^2,\overline{\sigma}^2])$ to represent the family
 $\{F_{\theta}\}_{\theta\in \Theta}$.

 Precisely,  the expectations of data $\{X_t\}_{0\leq t\leq T}$ under
 $\{P_{\theta}\}_{\theta\in \Theta}$ are
 \begin{equation}
   \label{eq:E_th}
   {E}_{\theta}[\phi(X_t)]=\int_{R}\phi(x)dF_{\theta}(x),
 \end{equation}
 where $\phi\in C_{l.Lip}(\mathbb{R},\mathbb{R})$ is a test function
 describing the statistic of the data  $X_t$  that we are interested
 in. With  VaR prediction in view, we concentrate our analysis on
 the {\em  worst-case expectation} of $X_t$ under
 $\{P_{\theta}\}_{\theta\in \Theta}$, that is,
 \begin{equation}
   \label{eq:EE}
   \mathbb{E}[\phi(X_t)]=\sup_{\theta\in \Theta}E_{\theta}[\phi(X_t)].
 \end{equation}
 In general, it is difficult to determine   worst-case expectation
 $ \mathbb{E}[\phi(X_t)]$. There is, however, a situation in which  it can
 be explicitly determined, as shown below.

 \begin{assumption} \label{ass:SDE}
   Suppose that $\{X_t\}_{0<t<T}$ satisfies  the following stochastic differential equation,
   $$
   dX_t=\theta_tdB_t,\quad X_0=x,
   $$
   under $P_{\theta}$, $\theta\in \Theta=L^2_{\gai{\cal F}}(\Omega\times[0,T],[\underline{\sigma},\overline{\sigma}])$.
 \end{assumption}
 For a given $\phi$, we can prove that $ u(t,x)=\mathbb{E}[\phi(X_t)]$
 satisfies the following partial differential equation (PDE),
 {
 \begin{equation}
   \label{pde-1}
   \partial_tu(t,x)-G(\partial_{xx}^2u)=0,\ t\geq 0,\ x\in \mathbb{R},
 \end{equation}
 }
{  where \gai{the} function $G(\cdot)$ is defined as
\begin{equation}
  \label{eq:G-function}
  G(a) = \frac12  \left( \overline{\sigma}^2 a^+ -
    \underline{\sigma}^2a^-   \right),
  \quad a^+=\max(a,0), \quad \text{and} \quad a^-=\max(-a,0).
\end{equation}}If, in addition,
 $\phi(\cdot)$ is convex on $\mathbb{R}$,  then the following explicit
 solution exists for the equation (\ref{pde-1}):
 $$
 u(t,x)=\int_{-\infty}^x\phi(y)\frac{1}{\sqrt{2\pi t \overline{\sigma}^2}}\exp(\frac{y^2}{2t\overline{\sigma}^2})dy.
 $$
 In this case, we can see that the convex function $u(t,x)$ will reach
 the maximum at parameter $\overline{\sigma}$. Similarly, if
 $\phi(\cdot)$ is concave on $\mathbb{R}$, then the explicit solution becomes
  $$
 u(t,x)=\int_{-\infty}^x\phi(y)\frac{1}{\sqrt{2\pi t \underline{\sigma}^2}}\exp(\frac{y^2}{2t\underline{\sigma}^2})dy.
 $$

 As a time-series of returns $\{X_t\}_{0\leq t\leq T}$ is typically
 centered { whereas }  its volatility (variance) is time-varying, we will
 hereafter assume that it satisfies Assumption~\ref{ass:SDE} under
 model uncertainty $\{P_\theta\}_{\theta\in\Theta} $ and follows a
 G-normal
 distribution $N(0,[\underline{\sigma}^2,\overline{\sigma}^2])$ (the
 reader is reminded that this is {\em not} a classical probability
 distribution). The definitions of maximal distribution and unbiased estimator are given as follows:

{
  \begin{definition}
    (Maximal distribution) A random \gai{variable}  $\eta$ on a
    sublinear expectation space $(\Omega,\mathcal{H},\mathbb{E})$ is
    called maximally distributed if there exists \gai{an} interval  $[\underline{\mu},\overline{\mu}] \subset \mathbb{R}$ such that
$$
\displaystyle \mathbb{E}[\phi(\eta)]=\max_{y\in [\underline{\mu},\overline{\mu}]}\phi(y),\quad \forall \phi\in C_{l.Lip}(\mathbb{R}).
$$
\end{definition}

\begin{definition}
Let $X_1,\cdots ,X_n$ be i.i.d. sample (under SLE) of size $n$ from the
maximal distribution with interval
$[\underline{\mu},\overline{\mu}]$. We say $f(X_1,\cdots ,X_n)$ is an
unbiased estimator of $\underline{\mu}$ or $\overline{\mu}$ if
$\mathbb{E}[f(X_1,\cdots ,X_n)]=\underline{\mu}$ or
$\mathbb{E}[f(X_1,\cdots ,X_n)]=\overline{\mu}$, respectively.
\end{definition}
}


%
\section{G-VaR: a new VaR approach under model uncertainty}
\label{sec:GVaR}

First  recall that,  given $\alpha\in
(0,1)$, the  $\text{VaR}_{\alpha}$ at the risk level
$\alpha$ of a financial asset $ X$ is the negative of the
level-$\alpha$ quantile  of
$X$;  that is
\begin{equation}\label{eq:VaR}
  \text{VaR}_{\alpha}(X) =-\inf\{x~:~ F(x)> \alpha\},
\end{equation}
where $F(x)=P(X\leq x)$ is the cumulative distribution function of
$X$.

\subsection{Robust VaR}

Consider a risky  position $X$ under model uncertainty represented by a family of   distributions
$\{F_{\theta}(x)\}_{\theta\in \Theta}$.
The VaR  of $X$ under each $F_{\theta}$ is
$$
\text{VaR}^{\theta}_{\alpha}(X)=-\inf\{x~:~ F_{\theta}(x)> \alpha\}.
$$

Under the  distribution family considered here, it is important to design
a VaR measure that can protect itself against  risk. Note that
 risk here takes a quite general form; that is, no specific
form or prior information is available on this family of  distributions. This
generality is aligned with real market situations in which risk factors are
always difficult, and  perhaps impossible,  to capture precisely. Hence,
any particular form or modeling of the sources of these risks could be misleading.
Accordingly, it becomes natural to consider a {\em
  worst-case} scenario for  VaR.
Formally,  the  worst-case VaR  of $X$ is here defined as
\begin{equation}
  \label{eq:VaR*}
  \text{VaR}^*_{\alpha}(X):=\sup_{\theta\in \Theta}\text{VaR}^{\theta}_{\alpha}(X).
\end{equation}
In the empirical study in Section~\ref{sec:empirical}, it will be shown
that, despite its conservative spirit,  consideration of a  worst-case
scenario, allows the new  VaR  to capture the risks in asset returns very efficiently.

The worst-case  VaR~\eqref{eq:VaR*} has several simple properties.
Let
\begin{equation}
  \label{eq:Fhat}
  \hat{F}(x):=\sup_{\theta\in \Theta}F_{\theta}(x).
\end{equation}
For each $\theta$, we have
$$
\text{VaR}^{\theta}_{\alpha}(X)\leq -\inf\{x~:~\hat{F}(x)> \alpha\}=: \text{VaR}^{\hat{F}}_{\alpha}(X),
$$
such that  $\text{VaR}^*_{\alpha}(X)\le\text{VaR}^{\hat{F}}_{\alpha}(X) $.

\begin{remark}
  It is clear that $\hat{F}$ is a right continuous function. If, in
  addition,  $\{F_{\theta}\}_{\theta\in \Theta}$ is weakly compact,then  it is easy to prove that
  $$
  \lim_{x\to -\infty} \hat{F}(x)=0, \quad  \text{and}\quad
  \lim_{x\to \infty}\hat{F}(x)=1.
  $$
  Thus in this case, $\hat{F}$ is still a probability distribution function.
\end{remark}

\begin{proposition}\label{prop-var}
  { Given a  risky position $X$ and a
  family of  distributions
  $\{F_{\theta}(x)\}_{\theta\in \Theta}$ that is weakly   compact.} Then,
  $$
  \text{VaR}^*_{\alpha}(X):=\sup_{\theta\in\Theta}\text{VaR}^{\theta}_{\alpha}(X)=\text{VaR}^{\hat{F}}_{\alpha}(X).
  $$
\end{proposition}

\begin{proof}
  It is clear that $\text{VaR}^*_{\alpha}(X)\leq \text{VaR}^{\hat{F}}_{\alpha}(X)$. To prove the reverse inequality, it suffices
  to find an $\tilde{F}\in \{F_{\theta}(x)\}_{\theta\in \Theta}$ such that
  $$
  \text{VaR}^{\hat{F}}_{\alpha}(X)=\text{VaR}^{\tilde{F}}_{\alpha}(X).
  $$
  Because $\hat{F}(x)$ is a right continuous non-decreasing function, we can find an $x_{\alpha}\in (-\infty,\infty)$ such that
  $$
  \hat{F}(x_{\alpha})\geq \alpha > \hat{F}(x),\  \text{for each}\  x< x_{\alpha}.
  $$
  Let $\{F_{\theta_{i}}\}_{i=1}^{\infty}$ be a subsequence of
  $\{F_{\theta}\}_{\theta\in \Theta}$ such that
  $F_{\theta_{i}}(x_{\alpha})\to \hat{F}(x_{\alpha})$. Because
  $\{F_{\theta}(x)\}_{\theta\in \Theta}$ is weakly compact, there exists a subsequence of $\{F_{\theta_{i_k}}\}_{k=1}^{\infty}$  such
  that $\{F_{\theta_{i_k}}\}_{k=1}^{\infty}$
  weakly converges to $\tilde{F}\in \{F_{\theta}(x)\}_{\theta\in \Theta}$. From
  $$
  \hat{F}(x_{\alpha})=\lim_{k\to \infty}F_{\theta_{i_k}}(x_{\alpha})\leq \tilde{F}(x_{\alpha})\leq \hat{F}(x_{\alpha}),
  $$
  it follows that $\tilde{F}(x_{\alpha})=\hat{F}(x_{\alpha})$. Further, for each $x< x_{\alpha}$, $\tilde{F}(x)\leq
  \hat{F}(x)<\hat{F}(x_{\alpha})$, namely,
  $$
  -\text{VaR}^{\hat{F}}_{\alpha}(X)=x_{\alpha}=\inf\{x~:~\tilde{F}(x)>\alpha\}=-\text{VaR}^{\hat{F}}_{\alpha}(X).
  $$
  The proof is complete.
\end{proof}

\subsection{G-VaR}

Based on Proposition \ref{prop-var}, we now introduce the concept of  G-VaR
under the family of distributions  $\{F_{\theta}\}_{\theta\in \Theta}$
for a risky asset $X$ that follows a G-normal distribution
$N(0,[\underline{\sigma}^2,\overline{\sigma}^2])$ depending on two positive
parameters $(\underline{\sigma},\overline{\sigma})$.
More precisely,  G-VaR is defined by  replacing the classical  distribution function
$F(x)$  in \eqref{eq:VaR} with  G-expectation $\mathbb{E}[1_{X\leq x}]$,
that is,
\begin{equation}
  \label{eq:GVaR}
  \text{G-VaR}_{\alpha}(X):=-\inf\{x\in \mathbb{R}:\mathbb{E}[1_{X\leq x}]> \alpha\}.
\end{equation}
{Note  that if  model uncertainty were absent}, the family
$\{F_\theta\}_{\theta\in\theta}$ would reduce to a single
distribution $F$, and  G-VaR would coincide with the traditional VaR
in \eqref{eq:VaR}.
Furthermore,  as $X$ follows a G-normal distribution, we have
$$
\hat{F}(x) = \mathbb{E}[1_{X\leq x}]= u(t,x)|_{t=1},
$$
where $u$ is the solution to  the nonlinear heat equation (\ref{pde-1}),
with Cauchy initial condition
\begin{equation}
  \label{pdecd-2}
\lim_{t\to 0}u(t,x)=1_{(0,\infty)}(x).
\end{equation}
By Proposition \ref{prop2},  function $\hat{F}$ has the following closed-form expressions.
$$
\hat{F}(x)=\int_{-\infty}^x\frac{\sqrt{2}}{(\overline{\sigma}+\underline{\sigma})  \sqrt{\pi }}
\left[ e^{ - \frac{y^2}{2\overline{\sigma}^2}  }   I(y\leq 0)+
  e^ { - \frac{y^2}{2\underline{\sigma}^2}} I(y> 0)  \right]dy,
$$
where $I(A)$ denotes the indicator function of a set $A$.
Moreover, evaluating the integral leads to the following more explicit
form of the function;
\begin{equation}
  \label{eq:Fhat2}
  \hat{F}(x)=
  \frac{2\overline{\sigma}}{\overline{\sigma}+\underline{\sigma}}
  \Phi(\frac{x}{\overline{\sigma}}) \  I(x\le 0)
  +\left\{
    1-    \frac{2\underline{\sigma}}{\overline{\sigma}+\underline{\sigma}}
    \Phi( -\frac{x}{\underline{\sigma}})
  \right\} \ I( x> 0),
\end{equation}
where $\Phi$ denotes the distribution function of the standard
normal.
This  G-normal distribution has a negative mean
$\sqrt{\frac{2}{\pi}}(\underline{\sigma}-\overline{\sigma})$, and a
negative skew. As an example, the  G-normal density function with parameters
$(\underline{\sigma},\overline{\sigma})=(0.5,1)$ is compared to the standard
normal density in Figure~\ref{gnorm1}.
Furthermore, as $\hat F $
is monotonically increasing, the G-VaR in~\eqref{eq:GVaR} is equal to
\begin{equation}
  \label{eq:GVaR2}
  \text{G-VaR}_{\alpha}(X) = - {\hat F}^{-1}(\alpha).
\end{equation}

\tabfig{Figure~\ref{gnorm1}}


{
  \subsection{A simple interpretation of G-VaR and some general comments}
  \label{norm-var}

  Although the G-VaR is  derived through the fairly sophisticated
  theory of SLE and related G-normal distribution,
  its final implementation shown in equations
  \eqref{eq:Fhat2} and \eqref{eq:GVaR2} has a  simple
  interpretation.
  First of all, a G-normal distribution we assumed for the asset $X$
  can be intuitively related to an  infinite family of normal
  distributions
  $\{ N(0,\sigma^2):\,  \sigma^2\in[{\underline\sigma}^2,{\overline\sigma}^2]\}$.
  Therefore the volatility interval
  $[{\underline\sigma}^2,{\overline\sigma}^2]$ \gai{corresponds} to  the model
  uncertainty considered here.
  As G-VaR is positive ($\alpha<0.5$), \eqref{eq:GVaR2} can be rewritten as a
 normal VaR:
  \begin{equation}\label{eq:normalGVaR}
    \text{G-VaR}_{\alpha}(X)=-\tilde{\sigma}\Phi^{-1}(\tilde \alpha),
  \end{equation}
  with
  \[
  \tilde\sigma = \overline{\sigma},\quad
  \tilde{\alpha}=\frac{\overline{\sigma}+\underline{\sigma}}{2\overline{\sigma}}\alpha
  = \frac12 \left( 1+ {\underline{\sigma}}/{ \overline{\sigma}} \right) \alpha.
  \]
  We will call  $\tilde\sigma$ and $ \tilde{\alpha}$ {\em adjusted
    volatility} and {\em adjusted risk level}, respectively.

  Therefore in absence of model uncertainty, one has
  $\tilde\sigma=\overline{\sigma}=\underline{\sigma}$ and
  $\tilde{\alpha}=\alpha$, the adjusted parameters coincide with the
  original ones  and G-VaR \gai{coincides} with  the traditional normal
  VaR. Otherwise, model uncertainty is characterized by the interval
  $[\underline{\sigma},  \overline{\sigma}]$ and the G-VaR becomes more
  conservative with a higher value, as it should be, under the joint
  effect of a larger value for the adjusted volatility $\tilde\sigma$
  and a smaller value of adjusted risk level $ \tilde{\alpha}<\alpha$.
  The degree of \gai{conservatism} is controlled by the two volatility
  parameters
  $\{\underline{\sigma},  \overline{\sigma}\}$: {\em larger model uncertainty
    leads to  more conservative G-VaR}.
  Note that this \gai{property} of G-VaR is intuitive given that VaR can be
  informally viewed as the maximum loss over a given period
  (after excluding   a given fraction of worst outcomes)\footnote{The
    authors thank a referee who recommended this valuable interpretation of G-VaR.}.

  In Section~\ref{img}, we will show how these  parameters are estimated from
  return data where a data-adaptive  window $W_0$ will eventually  determine the
  underlying volatility parameters
  $(\underline{\sigma},\overline{\sigma})$ for a  given risk level parameter $\alpha$.

  On a more methodological issue,
  one may ask whether it is reasonable to focus on the worst-case VaR among a family of distributions as we proposed in this study.
  To address the question, we think that given the inherent model
  uncertainty in  financial markets, a relevant  question is to seek
  for a "best"  aggregation of analyses from different  models.
  Our approach can be justified in two ways.
  Firstly from a theoretical perspective,
  the theory of sublinear expectation provides a
  rigourous derivation of the G-VaR under the worst-case scenario.
  Second,
  the data analyses in
  Section~\ref{sec:empirical} \gai{show} that G-VaR empirically perform well
  in comparison to existing benchmark VaR predictors.

  Indeed one can at some high level view G-VaR as 
  a particular {\em aggregation} method from the infinite normal models
  $\{ N(0,\sigma^2):\,  \sigma^2\in[{\underline\sigma}^2,{\overline\sigma}^2]\}$.
  It is thus interesting to ask ``whether there are other aggregation
  mechanisms that would be less conservative but still achieve a low
  number of violations''\footnote{The authors are grateful to a referee
    who suggested this discussion  with other aggregation methods.}.
  As the question is fairly general, it seems difficult to address it
  \gai{with} full precision. Instead we can here propose a particular  comparison
  with a recent method  for aggregation of \gai{misspecified} models
  proposed in \cite{GM18}.
}
{
  The paper   proposed a generalized aggregation approach for model averaging and
  applied it to asset pricing.
  This  method applies to  a finite number of ``misspecified models'' for
  aggregation.
  \gai{Note} that, we define the G-VaR model via a family of distributions
  $F_{\theta}(\cdot),\ \theta \in
  L^2(\Omega\times[0,T],[\underline{\sigma}^2,\overline{\sigma}^2])$
  which are related to an  infinite family of normal distributions $\{
  N(0,\sigma^2):\,
  \sigma^2\in[{\underline\sigma}^2,{\overline\sigma}^2]\}$. 
  It is thus natural to consider the averaging aggregation method  for a
  finite number, say $m$,  of normal distributions $N(0,\sigma_i^2)$,
  where $\sigma_i\in
  [\underline{\sigma},\overline{\sigma}],\ i=1,2,\cdots,m$.
  Let $F_i(\cdot)$ be  the cumulative distribution function (c.d.f.)
  of \gai{the}  normal
  distribution $N(0,\sigma_i^2)$ ($1\le  i\le m$).
  Implementing  the averaging aggregation mechanism in \cite{GM18} to
  VaR amounts to consider a c.d.f. of the form
  \begin{equation}\label{amm}
    \tilde{F}_{\omega}(x)=\bigg[\sum_{i=1}^m \omega_i(F_i(x))^{p}\bigg]^{1{/}p},\quad x\in \mathbb{R},
  \end{equation}
  where the weights $\omega=(\omega_1,\omega_2,\cdots,\omega_m)$ satisfy
  $\sum_{i=1}^m\omega_i=1,\ \omega_i\geq 0$ and $p>0$ is some ``shape'' parameter.
  Let  $F_0(\cdot)$ and $F_{m+1}(\cdot)$ be the c.d.f. of the two
  boundary normal distributions $N(0,\underline{\sigma}^2)$ and  $N(0,\overline{\sigma}^2)$, respectively.
  If $x<0$, we have
  $$
  F_0(x)\leq F_i(x)\leq F_{m+1}(x),\quad i=1,2,\cdots,m,
  $$
  and thus
  $$
  F_0(x)\leq \tilde{F}_{\omega}(x)\leq F_{m+1}(x).
  $$
  For \gai{any given risk} $\alpha<0.5$ and risky position $X$, it follows that,
  $$
  -\tilde{F}_{\omega}^{-1}(\alpha)\leq -F_{m+1}^{-1}(\alpha)\leq \text{G-VaR}_{\alpha}(X),
  $$
  where $\tilde{F}_{\omega}^{-1}(\cdot)$ and $F_{m+1}^{-1}(\cdot)$ are
  the inverse functions of $\tilde{F}_{\omega}(\cdot)$ and
  $F_{m+1}(\cdot)$, respectively. These results show that this 
  particular averaging aggregation method leads to a less conservative  VaR, but
  with a higher \gai{average} number of violations than the G-VaR.
  Note that the comparison here is quite particular and the
  averaging aggregation method is limited
  to a finite number of normal distributions. 
  Therefore, we still do not know \gai{whether} there exists any 
  other aggregation method for the VaR that is less conservative than
  G-VaR but with
  a lower number of violations.
}

\section{Implementation of G-VaR }\label{img}
\label{sec:implem}

In  implementing  G-VaR~\eqref{eq:GVaR},
the main task is to
estimate the parameters of the underlying G-normal distributions.
Let  $\{X_{t}\}_{0\leq t\leq T}$ be a return time-series from a risky
asset. At each time $t$, the goal is to forecast the VaR of  $X_{t+1}$
at a given level $\alpha$ using  the history of  available values $\{X_s\}_{0\le s\le t}$.
In the following, let window $W$ be the length of the trading days, which is used to estimate the parameters $(\underline{\sigma}^2,\overline{\sigma}^2)$.
To forecast the  VaR of  $X_{t+1}$ in the G-VaR model,
the data $\{X_{s+1}\}_{t-W}^{{t}}$, i.e., the history of
length $W$ before time $t$,  are assumed to be {\em independent and
identically distributed} (i.i.d.) with a  $G$-normal distribution,
$N(0,[\underline{\sigma}_{t}^{2},\overline{\sigma}_{t}^{2}])$. It is
important to remind the reader that the concepts of independence and
distribution equality used here are not the classical ones but those under
the theory of SLE
$\mathbb{E}[\cdot]:=\sup_{\theta\in \Theta}E_{\theta}[\cdot]$.
The appendix
provides a detailed introduction to these new concepts. Briefly,
under SLE,  two random variables $Y_1$ and $Y_2$ are
identically distributed if, for $\phi\in C_{l.Lip}(\mathbb{R})$,
$$
\mathbb{E}[\phi(Y_1)] = \mathbb{E}[\phi(Y_2)].
$$
A random variable $Y_2$ is said to be independent of $Y_1$ if, for
each $\psi\in C_{l.Lip}(\mathbb{R}\times\mathbb{R})$, we have
$$
\mathbb{E}[\psi(Y_1,Y_2)]=\mathbb{E}[\mathbb{E}[\psi(y_1,Y_2)]_{y_1=Y_1}].
$$
Note the ordering of  this independence: the fact that  $Y_2$ is
independent of $Y_1$ does not imply that  $Y_1$ is  independent of
$Y_2$. { In general, we \gai{say that } $Y_n$ is independent of $(Y_1,\cdots,Y_{n-1})$, if
$$
\mathbb{E}[\phi(Y_1,\cdots,Y_{n-1},Y_n)]=\mathbb{E}[\mathbb{E}[\phi(y_1,\cdots,y_{n-1},Y_n)]_{y_1=Y_1,\cdots,y_{n-1}=Y_{n-1}}],
$$
for \gai{any} $\phi\in C_{l.Lip}(\mathbb{R}^n)$.}

Theorem 24  in \cite{JP16} shows that if $X_1,\cdots,X_n$
form an  i.i.d. sample of size $n$ from a maximal distribution  with
parameters $(\underline{\mu},\overline{\mu})$, where the i.i.d. is under SLE, then
$$
\underline{\mu}\leq  \min\{X_1,\cdots ,X_n\} \leq   \max\{X_1,\cdots ,X_n\} \leq \overline{\mu}.
$$
Moreover,
$$
\hat{\overline{\mu}}=  \max\{X_1,\cdots ,X_n\}
$$
is the largest unbiased estimator for the upper mean $\overline{\mu}$, and
$$
\hat{\underline{\mu}}=  \min\{X_1,\cdots ,X_n\}
$$
is the smallest unbiased estimator for the lower mean
$\underline{\mu}$.  {This estimation method is called {``Max-Mean calculation''} in \cite{Peng2017}}.

Precisely, we assume that $X_{t+1}$ follows $N(0,[\underline{\sigma}_{t}^{2},\overline
{\sigma}_{t}^{2}]),\ t>0$. For each fixed $\bar{t}$, the
data $\{X_{\bar{t}-s}\}_{0\leq s\leq W-1}$ are used to estimate the
two parameters
$\underline{\sigma}_{\bar{t}}^{2}$ and
$\overline{\sigma}_{\bar{t}}^{2}$ for the  forecast of the  VaR of
$X_{\bar{t}+1}$.  Let $W_0\leq W$ be the  window width.  The following
moving window approach is then employed.
For each time $s$, let
\[
  \hat\sigma_{{s},W_0}^{2}
  =\hat{\sigma}_{s,W_0}^{2}(X_{{s}-W_0+1},\cdots,X_{{s}}) =\frac1{W_0}
  \sum_{j=1}^{W_0} X^2_{s-j+1}
\]
be the sample variance from the sample  $(X_{s-W_0+1},\ldots, X_s)$,
that is, the history of length $W_0$ before time $s$.
Let $k=\lfloor \frac{W}{W_0}\rfloor$ be the largest  integer satisfying $kW_0\leq W$.
Define
\begin{align*}
  & \hat{\overline{\sigma}}_{\bar{t},k}^{2} =\max \{\hat{\sigma}_{\bar{t}-s,W_0}^{2}:~  s=0,W_0,2W_0,\cdots,(k-1)W_0\},\\
  & \hat{\underline{\sigma}}_{\bar{t},k}^{2} =\min \{\hat{\sigma}_{\bar{t}-s,W_0}^{2}:~ s=0,W_0,2W_0,\cdots,(k-1)W_0\},\\
  &\hat{\overline{\sigma}}_{\bar{t},W_0}^{2}  =\max \{\hat{\sigma}_{\bar{t}-s,W_0}^{2}:~ 0\le s \le W-W_0\},\\
  &\hat{\underline{\sigma}}_{\bar{t},W_0}^{2} =\min \{\hat{\sigma}_{\bar{t}-s,W_0}^{2}:~0\le s \le W-W_0\}.\\
\end{align*}
{ \cite{Peng2010} shows that  the quadratic
variation process of a G-Brownian
motion follows a maximal distribution (see page 60 of the reference).  Thus, the quadratic variation
$\left\langle X \right\rangle_{t+1}$ follows a  maximal distribution
with \gai{interval} $[\underline{\sigma}^2_t,\overline{\sigma}^2_t]$.} By  Theorem 24  of \cite{JP16}, $ \hat{\overline{\sigma}}_{\bar{t},k}^{2} $  is the largest unbiased estimator for the upper mean $\overline{\sigma}_{\bar{t}}^{2}$ , and $ \hat{\underline{\sigma}}_{\bar{t},k}^{2} $  is the smallest unbiased estimator for the lower mean $\underline{\sigma}_{\bar{t}}^{2}$. { In addition, $ \hat{\overline{\sigma}}_{\bar{t},k}^{2}$ and $ \hat{\underline{\sigma}}_{\bar{t},k}^{2} $ converge to $\overline{\sigma}_{\bar{t}}^{2}$ and $\underline{\sigma}_{\bar{t}}^{2}$ as $W\to \infty$, respectively . Also note that
$$
\underline{\sigma}_{\bar{t}}^{2}\leq \hat{\underline{\sigma}}_{\bar{t},W_0}^{2} \leq \hat{\underline{\sigma}}_{\bar{t},k}^{2}
\leq
\hat{\overline{\sigma}}_{\bar{t},k}^{2} \leq \hat{\overline{\sigma}}_{\bar{t},W_0}^{2} \leq \overline{\sigma}_{\bar{t}}^{2},
$$
thus $\hat{\overline{\sigma}}_{\bar{t},W_0}^{2} $ and $ \hat{\underline{\sigma}}_{\bar{t},W_0}^{2} $ converge to $\overline{\sigma}_{\bar{t}}^{2}$ and $\underline{\sigma}_{\bar{t}}^{2}$ as $W\to \infty$.
These  shows that we can use  $\hat{\overline{\sigma}}_{\bar{t},W_0}^{2} $  and $ \hat{\underline{\sigma}}_{\bar{t},W_0}^{2} $ to estimate $\overline{\sigma}_{\bar{t}}^{2}$,
and $\underline{\sigma}_{\bar{t}}^{2}$, under
the given length of $W$ historical data and $W_0$}.\footnote{
  \cite{FPSS17}  provide a convergence  rate for these estimators  under sublinear expectation.}

In summary, at a given time $\bar t+1$, where a VaR forecast is
required,
by acknowledging model uncertainty in the historical data of length
$W$,  $\{X_t\}_{\bar t -W <t\le \bar t} $, $X_{\bar t+1}$ follows the
G-normal distribution $ N(0,[\underline{\sigma}^2_{\bar t},
\overline{\sigma}^2_{\bar t}  ])$. Moreover, the two parameters
$\underline{\sigma}^2_{\bar t}$ and $\overline{\sigma}^2_{\bar t}  $
can be well approximated by the estimators $\hat{\underline{\sigma}}_{\bar{t},W_0}^{2}$ and
$\hat{\overline{\sigma}}_{\bar{t},W_0}^{2}$, respectively.
Consequently, by \eqref{eq:GVaR2},   the final G-VaR estimate for the VaR of $X_{\bar t+1}$
at level $\alpha$ is
\begin{equation}
  \label{eq:GVaR-t1}
  \text{G-VaR}^{W_0}_{\alpha,\bar{t}}(X_{\bar{t}+1})=
  -  \left\{ {\hat   F}_{\bar t}^{W_0}\right\}^{-1}(\alpha),
\end{equation}
where
\begin{equation}\label{eq:GVaR-t2}
  {\hat    F}_{\bar t}^{W_0}(x)
  =
  \frac{2\hat{\overline{\sigma}}_{\bar{t},W_0}}{\hat{\overline{\sigma}}_{\bar{t},W_0}+\hat{\underline{\sigma}}_{\bar{t},W_0}}\Phi(\frac{x}{\hat{\overline{\sigma}}_{\bar{t},W_0}})
  ~I(x\le 0)
  +
  \left\{
    1-\frac{2\hat{\underline{\sigma}}_{\bar{t},W_0}}{\hat{\overline{\sigma}}_{\bar{t},W_0}+\hat{\underline{\sigma}}_{\bar{t},W_0}}\Phi(-\frac{x}{\hat{\underline{\sigma}}_{\bar{t},W_0}})
  \right\}
  ~ I(x> 0).
\end{equation}

Note that, for a given $W$ and $\alpha$, we  obtain different
$\hat{\overline{\sigma}}_{\bar{t},W_0}^{2} $ and
$\hat{\underline{\sigma}}_{\bar{t},W_0}^{2} $
\gai{depending on $W_0$}. Thus, $W_0$ can be interpreted as a measure of distribution
uncertainty.
\gai{This parameter $W_0$ plays a fundamental role in our analysis
  and we formalize its function in the following condition.
}

\begin{condition}\label{ass-2}
  {
    The series of returns   $\{X_t\}_{0\leq t\leq T}$ has the property
    that  for given  window $W$ and  risk level $\alpha$,
    there exists a \gai{window size}  $1\le W_0\le W$ such that
    \[
    \lim_{n\to \infty} \frac1{n-W}
        {\sum_{\bar{t}=W}^nI(X_{\bar{t}+1}<-\text{G-VaR}^{W_0}_{\alpha,\bar{t}}(X_{\bar{t}+1}))}
        =\alpha. 
        \] 
  }
\end{condition}

{
  \gai{When Condition~\ref{ass-2} is satisfied, we say that the return
  series has an {\em adaptive window} $W_0$ for a given pair
  $(\alpha,W)$. The condition}
  thus guarantees  coherence between the empirical
  percentages of violations and  G-VaR measure
  $\text{G-VaR}^{W_0}_{\alpha,\bar{t}}$ in \eqref{eq:GVaR-t1} for the
  entire dataset $\{X_t\}_{1\le t\le T}$ with historical window size
  $W$ and estimation window size $W_0$. Based on a large amount of
  data analysis,
  Condition~\ref{ass-2} appears fairly satisfied in many practical situations.}

 { As noted earlier, this paper  assumes that the observations are described  by infinitely many models (or distributions) rather than a single model.} Using
 the SLE theory and  adopting a worst-case
 scenario, the VaR at a given risk level $\alpha$ and time $t$
 can be evaluated through a  $G$-normal distribution
  $N(0,[\underline{\sigma}_{}^{2},\overline{\sigma}_{}^{2}])$.
  The parameter $W_0$ can be interpreted as the time duration  for which this
  worst-case scenario best fits the returns data. Moreover, as will
  be confirmed by the experiments in Section~\ref{ssec:W0}, this parameter
  depends on both the risk level $\alpha$ and  historical window
  size $W$.  { For example, at higher confidence levels (lower
    levels of $\alpha$), a higher degree of conservatism is reached by
    a lower value of $W_0$, which { generates} a greater volatility interval $[\underline{\sigma}^{2},\overline{\sigma}^{2}]$.}\footnote{%
    It has been reported in the literature that the parameters of a VaR model
    can depend on  the risk level $\alpha$. For example,
    \cite{Kuester06} observed that the normality assumption on the
    data ``might have some merit for larger values of $\alpha$,'' but
     is still not adequate for a 5\% risk level (see the second paragraph
    on page 76 \gai{of the reference}).
  } %
  Therefore, in real data analysis, such as that in
  Section~\ref{sec:empirical},
  for a given pair $(W,\alpha)$, we first check whether an adaptive
  window size $W_0$ exists, see Condition~\ref{ass-2}.  G-VaR
  forecasts are  possible only after \gai{finding such} $W_0$.
  Anticipating  the empirical study in
  Section~\ref{sec:empirical},  we will show that for the
  NASDAQ Composite Index and  S\&P500 Index,
  \gai{an adaptive window $W_0$ can indeed be found}
  for a wide range of risk levels
  $\alpha$. { However, for the CSI300 Index, which  we  also analyzed,
  no adaptive window
  size $W_0$ can be found for a reasonable historical window size $W$ and
  risk level $\alpha>5\%$. For a given risk level $\alpha$, as
  discussed in Section~\ref{norm-var},   $\text{G-VaR}_{\alpha}(X)$
  is equivalent to {a normal VaR}
  with \gai{an adjusted} variance parameter  \gai{$\tilde{\sigma}^2$}
  and downward adjusted risk  level
  $\displaystyle 0.5\alpha<\tilde{\alpha}\leq \alpha$.}
  These results indicate that the VaR  under normal distribution
  $N(0,\overline{\sigma}^2)$ within the risk level parameters interval
  $[0.5\alpha,\alpha]$ cannot cover the risk of the CSI300 Index when
  $\alpha>5\%$.
  This negative result
  can be interpreted  as that 
  \gai{
    the model uncertainty in CSI300 Index might be beyond the range
    that can be captured by the proposed G-VaR}.

\section{Empirical results  of G-VaR}
\label{sec:empirical}

In this section, the G-VaR  forecasts are evaluated
for the NASDAQ Composite Index and S\&P 500
Index.\footnote{Data  are downloaded from
  https://finance.yahoo.com/lookup. Recall that the two indexes are  market
  value-weighted portfolios comprising more than  5000 and 500
  selected
  stocks, respectively.}
Both indexes  comprise daily closing levels. The main steps are as follows.

\medskip\noindent\textbf{Step 1 - data preparation:} The
NASDAQ Composite Index is denoted by $\{Z_{1,t}\}$, running from
February 8,1971 to June 22, 2001, with a total of $N=7675$ observations of
percentages. The S\&P 500 Index is denoted by $\{Z_{2,t}\}$,
running from January 3, 2000 to February 7, 2018, with  a total of
$N=4550$ observations.
Their  daily log-returns are
$$
r_{i,t}=100(\ln Z_{i,t}-\ln Z_{i,t-1}), \quad i=1,2.
$$

\cite{Kuester06} found that when using a historical window  $W=1000$,
the best VaR predictions for the NASDAQ index are obtained by
AR-GARCH filtered modeling such as the recommended  AR-GARCH-Skewed-t or
AR-GARCH-Skewed-t-EVT models. The G-VaR predictor proposed in this paper is
compared with these two benchmarks, as well as with a more traditional
AR-GARCH-Normal predictor using standard normals for the filtered residuals.

\medskip\noindent\textbf{Step 2 - AR(1) filtering:} To carry out the   G-VaR prediction, we first filter the data with the following
AR(1) process; that is, the series   $r_{i,t},\ i=1,2$ satisfy
the model equations
\begin{equation}
  \begin{aligned}
    &r_{i,t}=a_ir_{i,t-1}+\epsilon_{i}
  \end{aligned}
\end{equation}
where  the $\epsilon_{i}$ follow G-normal distributions
$N(0,[\underline{\sigma}^2_{i},\overline{\sigma}^2_{i}])$, $i=1,2$,
respectively.

\medskip\noindent
\textbf{Step 3 - selection of historical and estimation window lengths
  $W$ and $W_0$:} The implementation of the G-VaR in
Section~\ref{sec:implem} requires the values of
the two window lengths  $W$ and $W_0$ for a given risk level $\alpha$.
Note that, in our G-VaR model, $W_0$ is dependent on $\alpha$
and $W$. Similarly to  \cite{Kuester06}, we consider three historical
windows, $W$=1000, 500, and 250.   The corresponding values of $W_0$
are selected empirically to ensure that Condition~\ref{ass-2} holds.

\subsection{  NASDAQ Composite  Index}
\newcommand{\tabincell}[2]{\begin{tabular}{@{}#1@{}}#2\end{tabular}}

We first compare the G-VaR model with the AR(1)-GARCH(1,1)-Normal,
AR(1)-GARCH(1,1)-Skewed-t, and AR(1)-GARCH(1,1)-Skewed-t-EVT VaR models.
For given windows $W$=1000, 500, and 250, we show how to determine a
window $W_0$ that satisfies Condition~\ref{ass-2}.  For example, for a given $W=1000,\alpha=0.01$, and time point $\bar{t}$, we calculate the G-VaR of $r_{1,\bar{t}}$ ( NASDAQ return) with different $W_0\leq W$. Then, we choose the $W_0$ that satisfies
$$
\lim_{n\to \infty}  \frac1{n-W}
\sum_{\bar{t}=W}^nI(r_{1,\bar{t}+1}<-\text{G-VaR}^{W_0}_{\alpha,\bar{t}}(r_{1,\bar{t}+1}))  =0.01.
$$
Note that the above percentage  is the violation rate of $r_{1,t}$
under the G-VaR model from time  $W+1$ to $n+1$, hereafter denoted as  \%Viol$(n)$.

Figure \ref{nvio} plots  the evolution of
\%Viol$(n)$ as time $n$ varies.  Here, $\alpha=0.01$ is used, but the findings are similar for
other values of $\alpha$.
For all window sizes $W=1000,500$, and $250$,  we find that when $3000\leq
n-W$, the violation rate becomes close to the target
$\alpha=0.01$. All three cases  use a well-calibrated value of
$W_0=350,~120,~75$.  In practice, as done here for
the NASDAQ Index, $W_0$ has to be calibrated, and, once fixed, it is
kept to forecast future G-VaR values. { Condition~\ref{ass-2}} technically
guarantees the existence of such a converging window size $W_0$.
Table~\ref{table:ns} presents the summary statistics for the \%Viol rates
over the converging period $3000<n-W$.

\tabfig{Table~\ref{table:ns}}

\tabfig{Figure~\ref{nvio}}

To assess the predictive performance of the models under
consideration, we follow the test of unconditional coverage, or the
binomial test  \citep{Kuester06}.
This is in fact a likelihood ratio test for a Bernoulli trial in which
the null trial success probability is equal to $\alpha$. More precisely,
let  $\hat{\alpha} = m_1/(m_0+m_1)$ be the sample violation rate
\%Viol, where  $m_1$ is the sample  number of violations, and   the total number of
observations is $m_0+m_1=T-W$.
Using the well-known asymptotic $\chi^2(1)$ distribution, the
$p$-value of the test is
$$
\text{LR}_{uc}=P\left(\chi^2(1)>2 m_1\frac{\hat{\alpha}}{\alpha}+2m_0\frac{1-\hat{\alpha}}{1-\alpha}\right).
$$

Table \ref{table:ns1000} gives the empirical values of the statistics  $\%
\text{Viol}, \text{LR}_{uc},  100\overline{\text{VaR}}$ for the  G-VaR
with $n=T$, $W=1000$, and $\alpha=0.01,0.025,0.05$. Here,
$100\overline{\text{VaR}}$ means $100$ times the average VaR of the
related model.   The  corresponding values of  $\% \text{Viol},
\text{LR}_{uc},  100\overline{\text{VaR}}$  for the
AR(1)-GARCH(1,1)-Normal, AR(1)-GARCH(1,1)-Skewed-t, and
AR(1)-GARCH(1,1)-Skewed-t-EVT models are directly imported  from Table 3 in \cite{Kuester06}. The results in  Table \ref{table:ns1000} show
that, for a given $W=1000, \alpha=0.01,0.025,0.05$, once we find the
corresponding $W_0=350,650,900$, the \%Viol of G-VaR is better than that of the
AR(1)-GARCH(1,1) model with  Normal, Skewed-t,  Skewed-t-EVT
innovations. See the plot of the $p$-value $\text{LR}_{uc}$ at the
bottom of the table. In addition, the values of
$100\overline{\text{VaR}}$ in the four models are very close to one another.

\tabfig{Table~\ref{table:ns1000}}


\cite{Kuester06} concluded  that the  AR-GARCH-Skewed-t and
AR-GARCH-Skewed-t-EVT VaR models achieve better performance with { larger}
windows, e.g.,
$W=1000$, than with smaller windows, e.g., $W=500,250$.
For G-VaR, however, as suggested by Figure~\ref{nvio}, the \%Viol
statistics are much more stable for  $W=500,250$, which
 suggests  better performance  with smaller
windows. To verify that suggestion, Table~\ref{table:ns500250}
gives the empirical statistics
of  $\% \text{Viol}, \text{LR}_{uc},  100\overline{\text{VaR}}$ from
the G-VaR model for windows $W=500,250$ and  risk levels
$\alpha$=0.003, 0.005, 0.01, 0.025, 0.05,  thereby confirming that G-VaR indeed achieves
excellent performance with  smaller windows. For the
difficult case with the lowest risk level, $\alpha=0.003$, the empirical
$p$-value even achieves top values of 0.96 and 1.00!

\tabfig{Table~\ref{table:ns500250}}

\subsection{ S\&P500 Index }
The  S\&P500 Index data are  analyzed using  windows
$W=1000,500,250$.
The window size $W_0$ in Condition~\ref{ass-2} is determined in exactly
 the same way as  for the NASDAQ Composite Index
data. Figure~\ref{svio} replicates Figure~\ref{nvio}, but for the
S\&P500 Index. Summary statistics of \%Viol from  $3000 <n-W$ are
given in Table~\ref{table:sp} (replicating Table~\ref{table:ns}, but for the
S\&P500 Index).

\tabfig{Table~\ref{table:sp}}

\tabfig{Figure~\ref{svio}}

Table \ref{table:sp1000} gives the empirical statistics of  $\% \text{Viol}, \text{LR}_{uc},  100\overline{\text{VaR}}$ for the AR(1)-GARCH(1,1)-Normal, AR(1)-GARCH(1,1)-Skewed-t, AR(1)-GARCH(1,1)-Skewed-t-EVT, and G-VaR models for the S\&P500 data with $n=T$, $W=1000$, and $\alpha=0.003,0.005,0.01,0.025,0.05$. These results  show
that, with a well-calibrated value for $W_0$,
G-VaR clearly outperforms  the three benchmark VaR predictors using
AR(1)-GARCH(1,1) filter and
Normal, Skewed-t, and Skewed-t-EVT innovations. See the
 $p$-value plot at the bottom of the table.

The experiments were then repeated for smaller windows, i.e., $W=500$ and 250. The
corresponding results for $W=500$ are given in
Table~\ref{table:sp500}, with calibrated values
$W_0=70,110,120,250,480$ for the various risk levels.
With the exception of  one case, namely, $\alpha=0.05$ with the GARCH-ST-EVT model,
G-VaR again  outperforms all  competitors; see the plot at the
bottom of the table.
The results for $W=250$ are reported in
Table~\ref{table:sp250}.
Here, the G-VaR outperforms all  competitors uniformly and
significantly.

The last plots in  Figures \ref{sw1} and \ref{sw5}  show the time
evolution of the three  one-step-ahead forecasts given
by  the G-VaR,  AR(1)-GARCH(1,1)-Skewed-t and AR(1)-
GARCH(1,1)-Skewed-t-EVT models.
Each plot comprises three  historical  windows,
$W=1000,500,250$.  The risk level is $\alpha=0.01$ in
Figure~\ref{sw1},  and $\alpha=0.05$ in Figure~\ref{sw5}.
All  three VaR predictors  have the capacity to follow the
rise-drop  patterns of the original return-series.

%
\tabfig{Table~\ref{table:sp1000}}

\tabfig{Table~\ref{table:sp500}}

\tabfig{Table~\ref{table:sp250}}

\tabfig{Figure~\ref{sw1}}

\tabfig{Figure~\ref{sw5}}

\subsection{Adaptive window size $W_0$}
\label{ssec:W0}

The implementation of   G-VaR forecasts requires the existence of
an adaptive window  $W_0$ that enables  estimation of the variance
parameters  $(\underline{\sigma}_t,\overline{\sigma}_t)$ of the
G-Normal distribution
$N(0,[\underline{\sigma}^2_{t},\overline{\sigma}^2_{t}])$ used at some
given time point $t$.  Figure~\ref{w0} displays the values of $W_0$
for the different values of $\alpha$  given in
Tables~\ref{table:sp1000}-\ref{table:sp250}.

{
  $W_0$  increases  with  risk level
  $\alpha$.  As explained earlier (see comments after
  Condition~\ref{ass-2}), a smaller $\alpha$  implies greater
  volatility, and   a smaller window $W_0$ is thus needed under the
  worst-case
  scenario adopted in this paper.}

\tabfig{Figure~\ref{w0}}

{
The information carried by these experimental
  values of $W_0$ can be pushed further. In  Figure~\ref{np500}, we compare the values
  found for both the NASDAQ Composite Index and S\&P500 Index
  under different risk levels $\alpha\in\{0.003,0.005,0.01,0.025,0.05\}$,
  with the  historical window size   fixed at $W=500$.
  Because under our worst-case scenario, smaller  window sizes  $W_0$
  correspond to higher volatility (and thus higher risk in the index),
  we can assume that the S\&P500 Index return is riskier than the
  NASDAQ Composite Index return at risk level $\alpha=0.025$. Their degree of risk
  is comparable  at risk level $\alpha=0.01$,
  and
  the S\&P500 Index is probably less risky
  at risk levels $\alpha=0.003,0.005$, and 0.05.}

\tabfig{Figure~\ref{np500}}

{  Finally, note that we used an initial segment of 3000 data
  points to determine the parameter $W_0$ in both cases of the  NASDAQ
  Composite Index and S\&P 500 Index. \gai{Once} $W_0$ is found, it is kept
  fixed and  used
  for  different moving windows $W=1000,500,250$. Therefore,
  the G-VaR forecasts can be considered as in-sample  forecasts for
  initial  $3000+W$ data points, while they are out-sample ones for
  the remaining  $n-(3000+W)$ data points.}

\section{Discussion}
\label{sec:concl}

This paper introduces  a new VaR predictor,  G-VaR,  for financial
return series.
Our methodology is based on the model-uncertainty principle that the volatility of
returns cannot be adequately  characterized by a single statistical distribution
or model. Rather,  an infinite family of distributions is
necessary for full characterization.
Considering the worst-case  volatility scenario
among these numerous potential distributions, and  using
the recent theory of  SLE,
we formally identity  G-VaR  through a new mathematical object
called $G$-normal distribution.
Extensive empirical analysis using  the NASDAQ Composite Index and
S\&P500 Index shows the G-VaR predictor to outperform many of the existing benchmark
predictors of  VaR. Its superiority  is particularly significant
for low risk levels, such as $\alpha=1\%$ or 0.5\%.

It is difficult to provide a completely clear  explanation for  the
surprising success of  G-VaR. Most likely, the concept of model
uncertainty has particular strength when considering the volatility of
returns. Such volatility is time-varying, and is reputed to be complex in nature, and thus
 the worst-case scenario approach taken by  G-VaR over all potential
volatility distributions  proves to be an excellent fit to the analysed
data. Judged by the empirical results presented herein, this model-uncertainty
approach appears more powerful than many of the existing approaches with
model certainty, wherein a unique statistical distribution is assumed for
the volatility process.

However, a number of unanswered question remain to be investigated
in future. In particular, the implementation of  G-VaR  depends on
an adaptive  window $W_0$. Although it has been shown that this ``tuning parameter'' can be
efficiently determined empirically for the two datasets analyzed in
this paper, it would be
worth investigating more its intrinsic or physical meaning.
It would  also be valuable to  analyze the performance of the G-VaR predictor on
other financial series to determine the extent to which the
worst-case scenario approach under model uncertainty remains successful.
More generally, it would be useful to explore
other
financial or even non-financial datasets in which
model uncertainty is unavoidable. The SLE  theory
could also provide new data analytic tools in the vein of the G-VaR
approach developed in this paper for the volatility of returns.

\section*{Acknowledgement and note on the genesis of G-VaR}

{
  The authors are grateful to numerous detailed  suggestions and
  comments from two reviewers
  and the editor that have led to significant improvements of the paper. 
  Specifically, the interpretation of G-VaR in terms of adjusted
  normal VaR  in (\ref{eq:normalGVaR}) is proposed by a reviewer.
  The interesting comparison with an averaging aggregation of
  misspecified VaRs in Section~\ref{norm-var} is proposed by another
  reviewer.
}  

This work on G-VaR has benefited from many discussions the authors had at
various meetings and workshops organized mainly at Zhongtai Security
Institute of Finance, Shandon University. 
It was in a team led by S. Peng in this institute that  G-VaR was
conceived in spring 2015. 
A first reporting on G-VaR  with
some data analysis examples on CSI300 Index
is proposed by members of this team
at the Industrial Problem Solving Workshop in Finance  in Shanghai
Jiao Tong University in April 2015.
By May 2015,   two of the authors (S. Peng and S. Yang) drafted a
working paper (unpublished) on
G-VaR where the concept was formally defined using the sublinear
expectation theory.
Later in October 2015    Zhongtai Security Institute of Finance
produced a report for CFFX where G-VaR is used to analyze
some  data sets  from Chinese markets.
Another internal report of Zhongtai Security
Institute of Finance by Gong, Yang, Hu and Zhang in November 2015
proposed some experiments of G-VaR for Standard \& Poor 500 Index
data.
It is however important to note that  all these preliminary works are
mostly empirical without any theory of G-VaR and using some
rudimentary implementation of G-VaR.
It is in this paper that  the G-VaR concept is rigorously constructed and
justified. In addition,  all the
empirical studies in this paper are new and different of those
reported in the aforementioned internal reports. { In addition,
  we would like to thank L. S. Jiang and X. Y. Yue for helping with
  the solution of  the fully nonlinear PDE (\ref{pde-1}).}

\bibliography{gexp,var}

\begin{thebibliography}{38}
\providecommand{\natexlab}[1]{#1}
\providecommand{\url}[1]{\texttt{#1}}
\expandafter\ifx\csname urlstyle\endcsname\relax
  \providecommand{\doi}[1]{doi: #1}\else
  \providecommand{\doi}{doi: \begingroup \urlstyle{rm}\Url}\fi

\bibitem[Abad et~al.(2014)Abad, Benito, and Lipez]{Abad14}
P.~Abad, S.~Benito, and C.~Lipez.
\newblock A comprehensive review of value at risk methodologies.
\newblock \emph{The Spanish Review of Financial Economics}, 12\penalty0
  (1):\penalty0 15--32, 2014.
\newblock ISSN 2173-1268.

\bibitem[Artzner et~al.(1999)Artzner, Delbaen, Eber, and Heath]{A99}
P.~Artzner, F.~Delbaen, J.-M. Eber, and D.~Heath.
\newblock Coherent measures of risk.
\newblock \emph{Mathematical Finance}, 9:\penalty0 203--228, 1999.

\bibitem[Avellaneda et~al.(1995)Avellaneda, Levy, and Par\'{a}s]{ALP95}
M.~Avellaneda, A.~Levy, and A.~Par\'{a}s.
\newblock Pricing and hedging derivative securities in markets with uncertain
  volatilities.
\newblock \emph{Applied Mathematical Finance}, 2:\penalty0 73--88, 1995.

\bibitem[Barone-Adesi et~al.(1999)Barone-Adesi, Giannopoulos, and
  Vosper]{BGV99}
G.~Barone-Adesi, K.~Giannopoulos, and L.~Vosper.
\newblock {V}a{R} without correlations for portfolios of derivative securities.
\newblock \emph{Journal of Futures Markets}, 19:\penalty0 583--602, 1999.

\bibitem[Barone-Adesi et~al.(2002)Barone-Adesi, Giannopoulos, and
  Vosper]{BGV02}
G.~Barone-Adesi, K.~Giannopoulos, and L.~Vosper.
\newblock Backtesting derivative portfolios with {F}iltered {H}istorical
  {S}imulation ({FHS}).
\newblock \emph{European Financial Management}, 8:\penalty0 31--58, 2002.

\bibitem[Bollerslev(1986)]{Bollerslev86}
T.~Bollerslev.
\newblock Generalized autoregressive conditional heteroskedasticity.
\newblock \emph{Journal of Econometrics}, 31:\penalty0 307--327, 1986.

\bibitem[Chen and Epstein(2002)]{CE02}
Z.~Chen and L.~Epstein.
\newblock Ambiguity, risk, and asset returns in continuous time.
\newblock \emph{Econometrica}, 70\penalty0 (4):\penalty0 1403--1443, 2002.

\bibitem[Chernozhukov and Umantsev(2001)]{Chernozhukov01}
V.~Chernozhukov and L.~Umantsev.
\newblock Conditional value-at-risk: Aspects of modeling and estimation.
\newblock \emph{Empirical Economics}, 26\penalty0 (1):\penalty0 271--292, 2001.

\bibitem[Christoffersen(2003)]{Christoffersen03}
P.~F. Christoffersen.
\newblock \emph{Elements of Financial Risk Management}.
\newblock Academic Press, Amsterdam, 2003.

\bibitem[Cont(2006)]{C06}
R.~Cont.
\newblock Model uncertainty and its impact on the pricing of derivative
  instruments.
\newblock \emph{Mathematical Finance}, 16:\penalty0 519--547, 2006.

\bibitem[Denis et~al.(2011)Denis, Hu, and S.]{DHP11}
L.~Denis, M.~Hu, and Peng. S.
\newblock Function spaces and capacity related to a sublinear expectation:
  application to \text{G}-{B}rownian motion paths.
\newblock \emph{Potential Analysis}, 34:\penalty0 139--161, 2011.

\bibitem[Dowd(2002)]{Dowd02}
K.~Dowd.
\newblock \emph{Measuring Market Risk}.
\newblock John Wiley \& Sons, Chichester, 2002.

\bibitem[Embrechts et~al.(1997)Embrechts, Kl{\"u}ppelberg, and Mikosch]{EKM97}
P.~Embrechts, K.~Kl{\"u}ppelberg, and T.~Mikosch.
\newblock \emph{Modelling Extremal Events for Insurance and Finance}.
\newblock Springer, Berlin, 1997.

\bibitem[Engle and Manganelli(2004)]{Engle04}
R.F. Engle and S.~Manganelli.
\newblock \text{CAViaR}: Conditional autoregressive value at risk by regression
  quantiles.
\newblock \emph{Journal of Business and Economic Statistics}, 22\penalty0
  (4):\penalty0 367--381, 2004.

\bibitem[Epstein and Ji(2013)]{EJ13}
L.~G. Epstein and S.~Ji.
\newblock Ambiguous volatility and asset pricing in continuous time.
\newblock \emph{Review of Financial Studies}, 26\penalty0 (7):\penalty0
  1740--1786, 2013.

\bibitem[Fang et~al.(2019)Fang, Peng, Shao, and Song]{FPSS17}
X.~Fang, S.~Peng, Q.~Shao, and Y.~Song.
\newblock Limit theorems with rate of convergence under sublinear expectations.
\newblock \emph{Bernoulli}, 25:\penalty0 1--31, 2019.

\bibitem[F\"{o}llmer and Schied(2011)]{FS11}
H.~F\"{o}llmer and A.~Schied.
\newblock \emph{Stochastic Finance. An Introduction in Discrete Time. Third
  revised and extended edition}.
\newblock Walter de Gruyter \& Co., Berlin, 2011.

\bibitem[Gospodinov and Maasoumi(2018)]{GM18}
N.~Gospodinov and E.~Maasoumi.
\newblock Generalized aggregation of misspecified models: with an application
  to asset pricing.
\newblock Technical report, https://sites.google.com/site/gospodinovfed/, 2018.

\bibitem[Haas et~al.(2004)Haas, Mittnik, and Paolella]{Haas04}
M.~Haas, S.~Mittnik, and M.~S. Paolella.
\newblock Mixed normal conditional heteroskedasticity.
\newblock \emph{Journal of Financial Econometrics}, 2:\penalty0 211--250, 2004.

\bibitem[Huber(1981)]{Huber81}
P.~J. Huber.
\newblock \emph{Robust Statistics. Wiley Series in Probability and Mathematical
  Statistics}.
\newblock John Wiley \& Sons, Inc., New York, 3rd edition, 1981.

\bibitem[Jin and Peng(2016)]{JP16}
H.~Jin and S.~Peng.
\newblock Optimal unbiased estimation for maximal distribution.
\newblock \emph{arXiv:1611.07994v1}, 2016.

\bibitem[Jorion(2007)]{Jorion07}
Ph. Jorion.
\newblock \emph{\text{Value at Risk} : the New Benchmark for Managing Financial
  Risk}.
\newblock McGraw-Hill, New York, 3rd edition, 2007.

\bibitem[Jorion(2010)]{Jorion10}
Ph. Jorion.
\newblock Risk management.
\newblock \emph{Annual Review of Financial Economics}, 2\penalty0 (1):\penalty0
  347--365, 2010.

\bibitem[Kerkhof et~al.(2010)Kerkhof, Melenberg, and Schumacher]{KMS10}
J.~Kerkhof, B.~Melenberg, and H.~Schumacher.
\newblock Model risk and capital reserves.
\newblock \emph{Journal of Banking \& Finance}, 34:\penalty0 267--279, 2010.

\bibitem[Koenker and Bassett(1978)]{Koenker78}
R.~Koenker and G.~Bassett.
\newblock Regression quantiles.
\newblock \emph{Econometrica}, 46:\penalty0 33--50, 1978.

\bibitem[Kuester et~al.(2006)Kuester, Mittnik, and Paolella]{Kuester06}
K.~Kuester, S.~Mittnik, and M.~S. Paolella.
\newblock \text{Value-at-Risk Prediction}: A comparison of alternative
  strategies.
\newblock \emph{Journal of Financial Econometrics}, 4\penalty0 (1):\penalty0
  53--89, 2006.

\bibitem[Lyons(1995)]{L95}
T.~J. Lyons.
\newblock Uncertain volatility and the risk-free synthesis of derivatives.
\newblock \emph{Applied Mathematical Finance}, 2:\penalty0 117--133, 1995.

\bibitem[McNeil and Frey(2000)]{McNeil00}
A.~J. McNeil and R.~Frey.
\newblock Estimation of tail-related risk measures for heteroscedastic
  financial time series: An extreme value approach.
\newblock \emph{Journal of Empirical Finance}, 7:\penalty0 271--300, 2000.

\bibitem[Nadarajah and Chan(2016)]{NadaChan16}
S.~Nadarajah and S.~Chan.
\newblock Estimation methods for value at risk.
\newblock In \emph{Extreme Events in Finance: A Handbook of Extreme Value
  Theory and its Applications}, pages 283--356. Wiley, 2016.
\newblock ISBN 9781118650318.

\bibitem[Peng(1997)]{P97}
S.~Peng.
\newblock \emph{Backward SDE and related g-expectation. Backward stochastic
  differential equations (Paris, 1995-1996) 141-159}.
\newblock Pitman Res. Notes Math. Ser., Longman, Harlow, 1997.

\bibitem[Peng(2004)]{Peng2004}
S.~Peng.
\newblock Filtration consistent nonlinear expectations and evaluations of
  contingent claims.
\newblock \emph{Acta Mathematicae Applicatae Sinica}, 20:\penalty0 1--24, 2004.

\bibitem[Peng(2006)]{Peng2006}
S.~Peng.
\newblock \emph{Stochastic Analysis and Applications: The Abel Symposium 2005},
  chapter "$G$-expectation, $G$-{B}rownian motion and related stochastic
  calculus of {I}t\^o type".
\newblock Springer, Berlin Heidelberg, 2006.

\bibitem[Peng(2008)]{Peng2008}
S.~Peng.
\newblock Multi-dimensional \text{G}-{B}rownian motion and related stochastic
  calculus under \text{G}-expectation.
\newblock \emph{Stochastic Processes and Their Applications}, 118:\penalty0
  2223--2253, 2008.

\bibitem[Peng(2017)]{Peng2017}
S.~Peng.
\newblock Theory, methods and meaning of nonlinear expectation theory.
\newblock \emph{Scientia Sinica Mathematica: In Chinese}, 47:\penalty0
  1223--1254, 2017.

\bibitem[Peng(2019)]{Peng2010}
S.~Peng.
\newblock \emph{Nonlinear {E}xpectations and {S}tochastic {C}alculus {u}nder
  {U}ncertainty}.
\newblock Springer, Berlin, Heidelberg, 2019.

\bibitem[Pickands(1975)]{Pickands75}
J.~Pickands.
\newblock Statistical inference using extreme order statistics.
\newblock \emph{Annals of Statistics}, 3:\penalty0 119--131, 1975.

\bibitem[Walley(1991)]{Walley91}
P.~Walley.
\newblock \emph{Statistical reasoning with imprecise probabilities. Monographs
  on Statistics and Applied Probability, 42}.
\newblock Chapman and Hall, Ltd., London, 1991.

\bibitem[Zhang and Nadarajah(2017)]{Zhang17}
Y.~Zhang and S.~Nadarajah.
\newblock A review of backtesting for value at risk.
\newblock \emph{Communications in Statistics - Theory and Methods}, pages
  1--24, 2017.

\end{thebibliography}

\appendix

\section{Relevant results from   sublinear expectations}
\label{sec:sublinearE}

In this appendix, we introduce the relevant concepts and properties of
the general theory of sublinear expectation used in the paper. Let
$\Omega$ be an arbitrarily given set and $\mathcal{H}$ be a linear space of real functions, called random variables,
defined on $\Omega$ such that, if $\xi\in \mathcal{H}$, then
$\left|\xi \right|\in \mathcal{H}$. We also assume that $1\in
\mathcal{H}$. The space $\mathcal{H}$ is called a vector lattice on
$\Omega$. We make the following assumption: if $\xi_1,\cdots,\xi_n\in
\mathcal{H}$, or equivalently,
$\xi=(\xi_1,\cdots,\xi_n)\in \mathcal{H}^n$,  then $\phi(\xi)\in \mathcal{H}$ for each function $\phi$ in
$C_{l.Lip}(\mathbb{R}^n)$. Here, $\phi$ corresponds to some  characteristic of $\xi$, and $C_{l.Lip}(\mathbb{R}^n)$ is the space of all functions $\phi$ defined on $\mathbb{R}^n$ satisfying
$$
\left|\phi(x)-\phi(y)  \right|\leq C(1+\left|x \right|^m+\left|y
\right|^m)\left|x-y \right|, \quad x,y\in \mathbb{R}^n,
$$
for  some $C>0$ and  $m\in \mathbb{N}$ depending on $\phi$. Similarly, with the probability space, we introduce a sublinear expectation $\mathbb{E}$ on $\mathcal{H}$ that was first proposed in \cite{Peng2006}.

\begin{definition}
A function $\mathbb{E}:\mathcal{H}\to \mathbb{R}$ is called a sublinear expectation
on $(\Omega,\mathcal{H})$ if it satisfies

1). Monotonicity: if $X(\omega)\geq Y(\omega)$ for each $\omega\in
\Omega$, then $\mathbb{E}[X]\geq \mathbb{E}[Y]$;

2). $\mathbb{E}[X+c]=\mathbb{E}[X]+c$ for any $c\in\mathbb{R}$;

3). $\mathbb{E}[X+Y]\leq \mathbb{E}[X]+\mathbb{E}[Y]$; and

4). $\mathbb{E}[\lambda X]=\lambda \mathbb{E}[X]$ for any $\lambda\ge 0$.
\end{definition}

{ A closely associated concept is
  {\em coherent risk measures} introduced in \cite{A99}.
  From a \gai{mathematical} point of view,
  this concept essentially coincides with that of sublinear
  expectation.
  What makes the difference is perhaps the
  respective contexts where the two concepts  are introduced and used.
  Coherent risk measures are introduced within mathematical finance
  (or   financial mathematics) and  dedicated to study problems and
  subjects in this area, for example market risks and non-market risks.
  On the other hand,
  sublinear expectation is more abstract  and introduced within general
  probability theory.
  For example,  objects and concepts
  like
  Brownian motion, conditional expectation and  martingale
  have been  developed  under sublinear expectation.
  A related discussion is
  \cite{DHP11} where
  a  representation theorem of a coherent risk measure is proposed.
}

Let $X_1$ and $X_2$ be two $n$-dimensional random vectors defined on nonlinear
expectation spaces $(\Omega_1,\mathcal{H}_1,\mathbb{E}_1)$ and $(\Omega_2,\mathcal{H}_2,\mathbb{E}_2)$, respectively. They are called
identically distributed, denoted by $X_1:\overset{d}{=}X_2$, if
$$
\mathbb{E}_1[\phi(X_1)] = \mathbb{E}_2[\phi(X_2)],   \quad \forall \phi\in C_{l.Lip}(\mathbb{R}^n).
$$

A random vector $Y\in \mathcal{H}^n$ is said to be independent of $X\in \mathcal{H}^m$ if, for
each $\phi\in C_{l.Lip}(\mathbb{R}^m\times\mathbb{R}^n)$, we have
$$
\mathbb{E}[\phi(X,Y)]=\mathbb{E}[\mathbb{E}[\phi(x,Y)]_{x=X}].
$$
If the above equality holds only for a specific $\phi$, then we say that $Y$ is uncorrelated
to $X$ with respect to this function $\phi$. Under a sublinear expectation $\mathbb{E}$, the independence of $Y$ from $X$ means that the uncertainty of distributions of $Y$ does
not change with each realization of $X(\omega)=x,\ x\in\mathbb{R}^n$. It is important to note that
{
  this independence under sublinear  expectation is not symmetric;  in
  particular
  ``$Y$ is independent of $X$" is not equivalent to
  ``$X$ is independent of $Y$." For illustration we consider an
  example that  generalizes the one in \cite{Peng2010}.
}

{
\begin{example}
Let \gai{$\xi,X,Y,Z\in \mathcal{H}$ be identically distributed
  satisfying}
$\mathbb{E}[\xi]=\mathbb{E}[-\xi]=0$ and
${\overline{\sigma}}^2=\mathbb{E}[\xi^2]>{\underline{\sigma}}^2
=-\mathbb{E}[-\xi^2]$. We suppose $\mathbb{E}[\left|\xi \right|]
=\mathbb{E}[\xi^++\xi^-]>0$, thus
$$
\mathbb{E}[\xi^+]=\frac{1}{2}
\mathbb{E}[\left| \xi\right|+\xi]=\frac{1}{2}\mathbb{E}[\left| \xi\right|]>0,
$$
and
$$
\mathbb{E}[\xi^-]=\frac{1}{2}
\mathbb{E}[\left| \xi\right|-\xi]=\frac{1}{2}\mathbb{E}[\left| \xi\right|]>0.
$$
Assume that  $Z$ is independent of $\{X,Y\}$ and $Y$ is
independent of   $X$, we can show that
$$
\mathbb{E}[XY]=0,\ \mathbb{E}[\left|XY\right|]=(\mathbb{E}[\left|X\right|])^2>0,\
\mathbb{E}[(XY)^+]=\mathbb{E}[(XY)^-]=\frac{1}{2}\mathbb{E}[\left| XY\right|]>0,
$$
and
$$
\mathbb{E}[XYZ^2]=\mathbb{E}[(XY)^+{\overline{\sigma}}^2-(XY)^{-}{\underline{\sigma}}^2]
=({\overline{\sigma}}^2-{\underline{\sigma}}^2)\mathbb{E}[(XY)^+]>0.
$$
In contrast, { if we assume that} $X$ is independent of $Y$ and
$\{X,Y\}$ are independent of  $Z$, we have
$$
\mathbb{E}[XYZ^2]=\mathbb{E}[\mathbb{E}[XYz^2]_{z=Z}]
=\mathbb{E}[Z^2]\mathbb{E}[XY]=0.
$$
{ Therefore, by reversing the order of independence, we obtain
  different values for $\mathbb{E}[XYZ^2]$. This shows the asymmetry of
  the concept of independence under sublinear expectation while
  the concept is  symmetric in  classical probability
  theory.}
\end{example}
}


Finally, further properties of the G-normal distribution are
rigorously established as follows.
\begin{proposition}\label{prop2}
  The solution of the fully nonlinear PDE (\ref{pde-1}) with Cauchy initial
condition (\ref{pdecd-2}) has the following explicit expression;
$$
u(t,x)=\int_{-\infty}^x\rho(t,y)dy,
$$
where $\rho(t,y)$ is a function on $\mathbb{R_+}\times\mathbb{R}$ defined by
$$
\rho(t,y)=\frac{\sqrt{2}}{ (\overline{\sigma}+\underline{\sigma})
  \sqrt{\pi t}}
\left[  e^{- \frac{y^2}{2\overline{\sigma}^2t} }
  I( y\leq 0) +  e^{-  \frac{y^2}{2\underline{\sigma}^2t   } } I(y> 0)
\right].
$$
\end{proposition}

{
\begin{proof}
 \gai{Using classical analysis of the heat equation, we can prove that} 
  $\displaystyle\lim_{t\to 0} u(t,x)=1_{(0,\infty)}(x)$ for each $x$. If $x<0$, we have $\rho(t,x)=\frac{\sqrt{2}}{ (\overline{\sigma}+\underline{\sigma})
  \sqrt{\pi t}}
 e^{- \frac{x^2}{2\overline{\sigma}^2t} }
  $ and  can obtain that
    \begin{equation}
    \begin{aligned}
    &\partial_tu(t,x)=\partial_x\rho(t,x)=\frac{-x}{ (\overline{\sigma}+\underline{\sigma})
  \sqrt{2\pi t^3}}
 e^{- \frac{x^2}{2\overline{\sigma}^2t} },\\
      &\partial_xu(t,x)=\rho(t,x)=\frac{\sqrt{2}}{ (\overline{\sigma}+\underline{\sigma})
  \sqrt{\pi t}}
 e^{- \frac{x^2}{2\overline{\sigma}^2t} },\\
      &\partial_{xx}^2u(t,x)=\partial_x\rho(t,x)=\frac{-\sqrt{2}x}{ (\overline{\sigma}+\underline{\sigma}){\overline{\sigma}}^2
  \sqrt{\pi t^3}}
 e^{- \frac{x^2}{2\overline{\sigma}^2t} }.
    \end{aligned}
  \end{equation}
  Thus, we can verify
  $$
  \partial_tu(t,x)-G(\partial_{xx}^2u(t,x))=0,\ x<0,t>0.
  $$
 In a similar manner, we have
   $$
  \partial_tu(t,x)-G(\partial_{xx}^2u(t,x))=0,\ x<0,t>0.
  $$
Note that $\partial_{xx}^2u(t,x)$ is continuous on $(0,\infty)\times \mathbb{R}$, and $\partial_{xx}^2u(t,0)=0, \partial_tu(t,0)=0$, thus, one obtains,
$$
\partial_tu(t,x)-G(\partial_{xx}^2u(t,x))=0,\ x=0,t>0.
$$
Consequently, $u(t,x)$ solves the PDE (\ref{pde-1}) on the entire
  $(0,\infty)\times \mathbb{R}$ in the sense of classical solution.
\end{proof}
}

\vskip2cm


\begin{figure}[H]
  \centering
  \includegraphics[width=3.8 in]{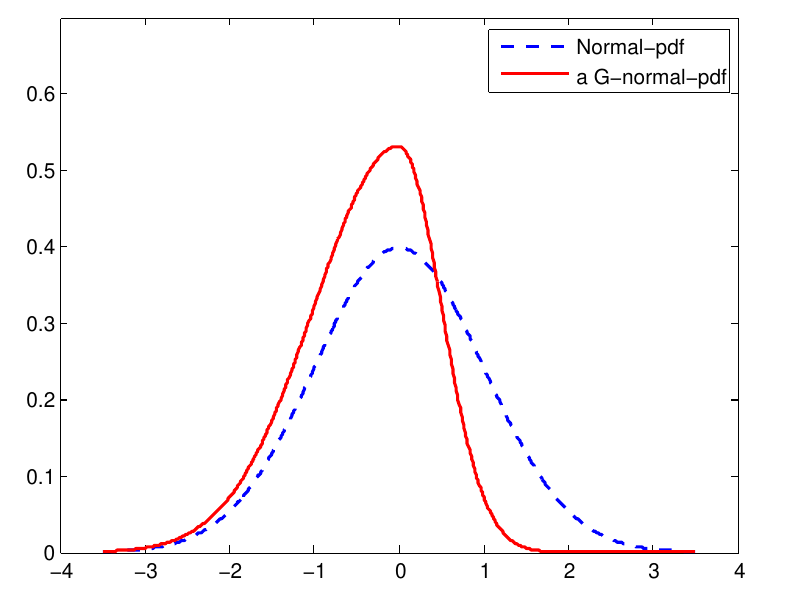}%
  \caption{\label{gnorm1}Density of one distribution of  G-normal distribution  with
    variance parameters $(\protect\underline{\sigma},\protect\overline{\sigma})=(0.5,1)$
    in comparison    with  standard    normal density.}
\end{figure}

\newpage\thispagestyle{empty}
\begin{figure}[H]
  \centering
  \includegraphics[width=4.2 in]{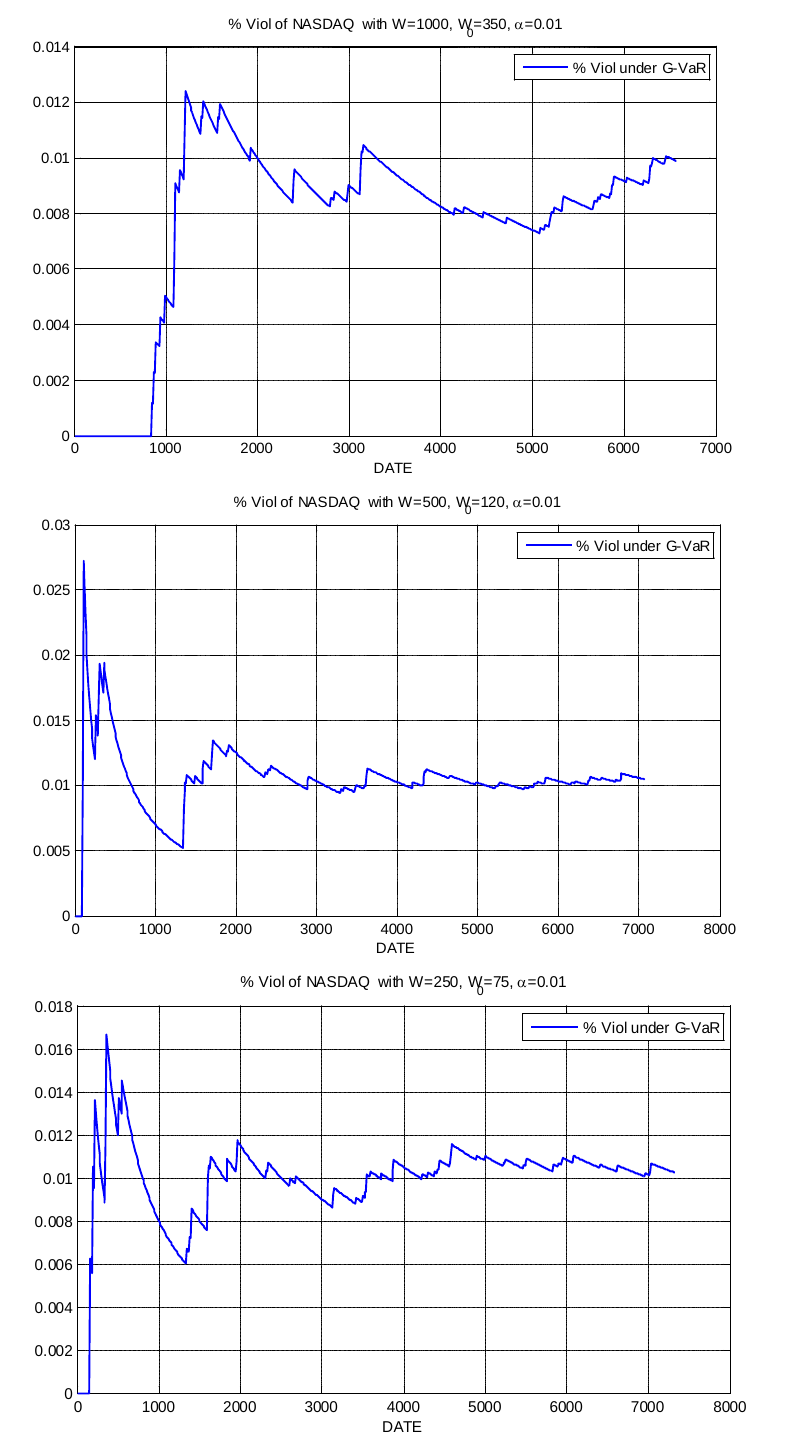}%
  \caption{NASDAQ Composite Index: convergence of the violation rate \%Viol
    for $W=1000,500,250$, and $\alpha=0.01$.
   The  adaptive window sizes are  $W_0=350,120,75$, respectively.  \label{nvio}}
\end{figure}

\newpage\thispagestyle{empty}
\begin{figure}[H]
  \centering
  \includegraphics[width=4.0 in]{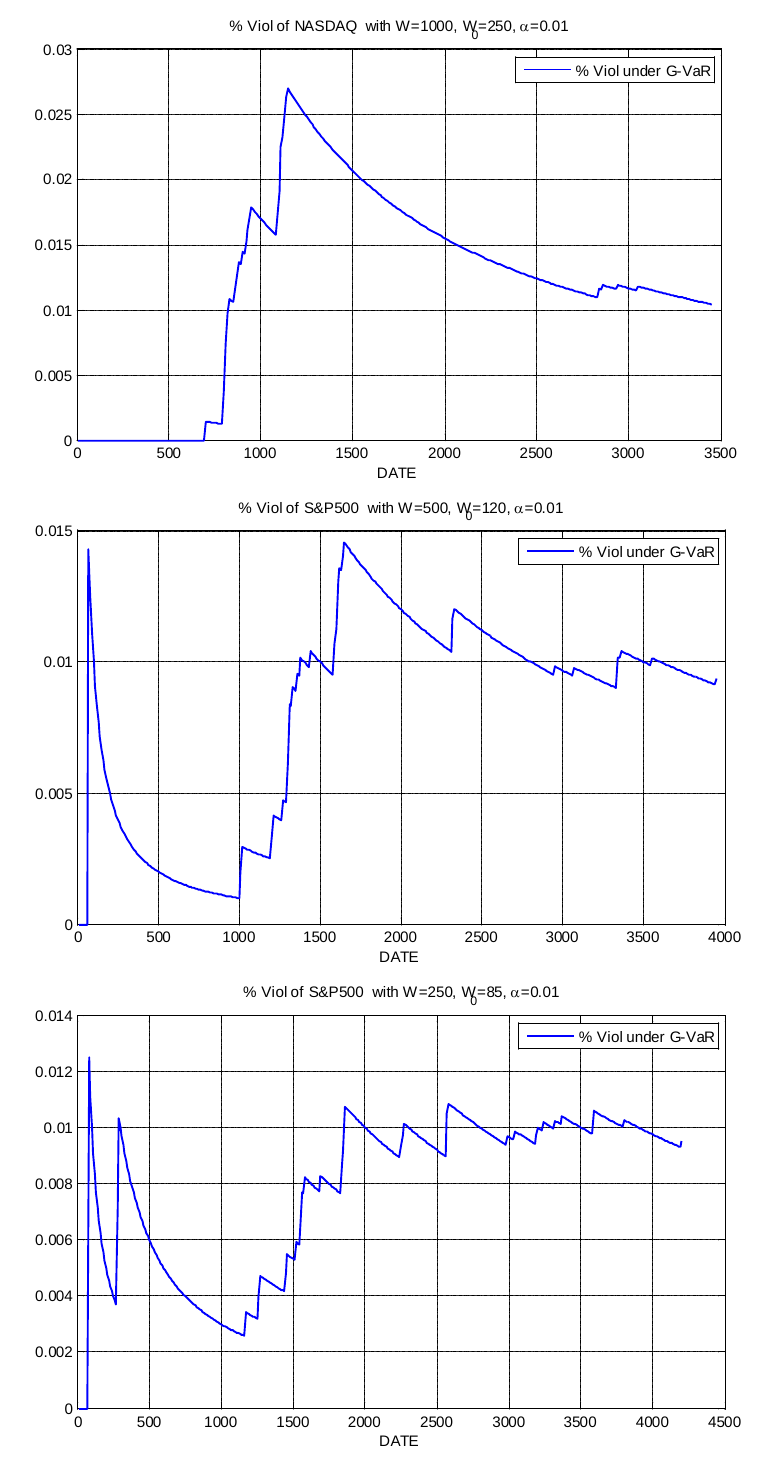}%
  \caption{S\&P500 Index: convergence of the violation rate \%Viol
    for $W=1000,500,250$,  and $\alpha=0.01$.
    These adaptive window sizes are  $W_0=250,120,85$, respectively.  \label{svio}}
\end{figure}

\newpage\thispagestyle{empty}
\begin{figure}[H]
  \centering
\includegraphics[width=3.8 in]{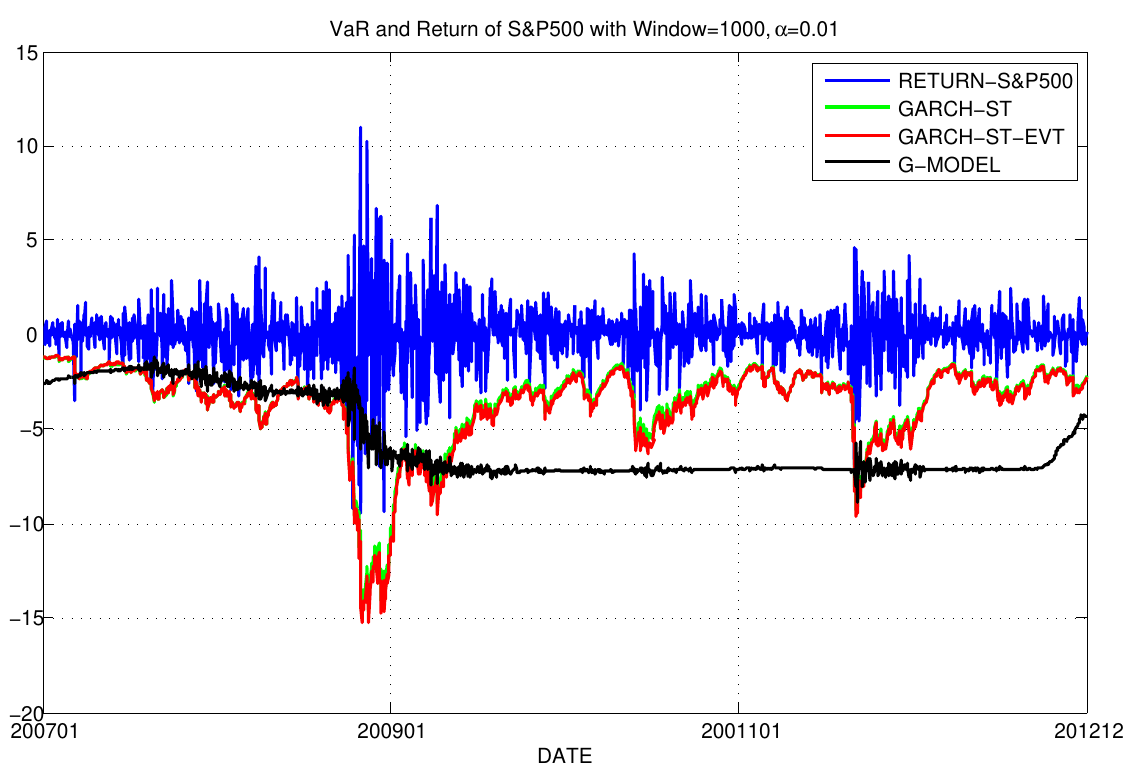}\\
\includegraphics[width=3.8 in]{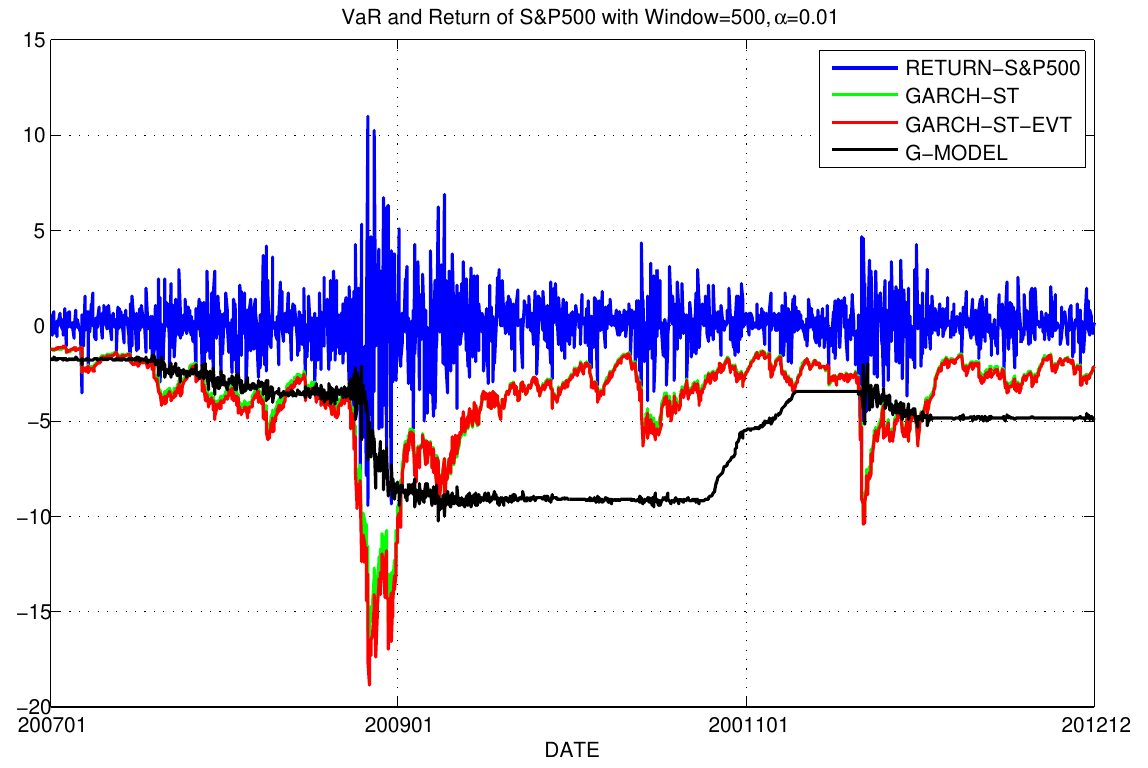}\\
\includegraphics[width=3.8 in]{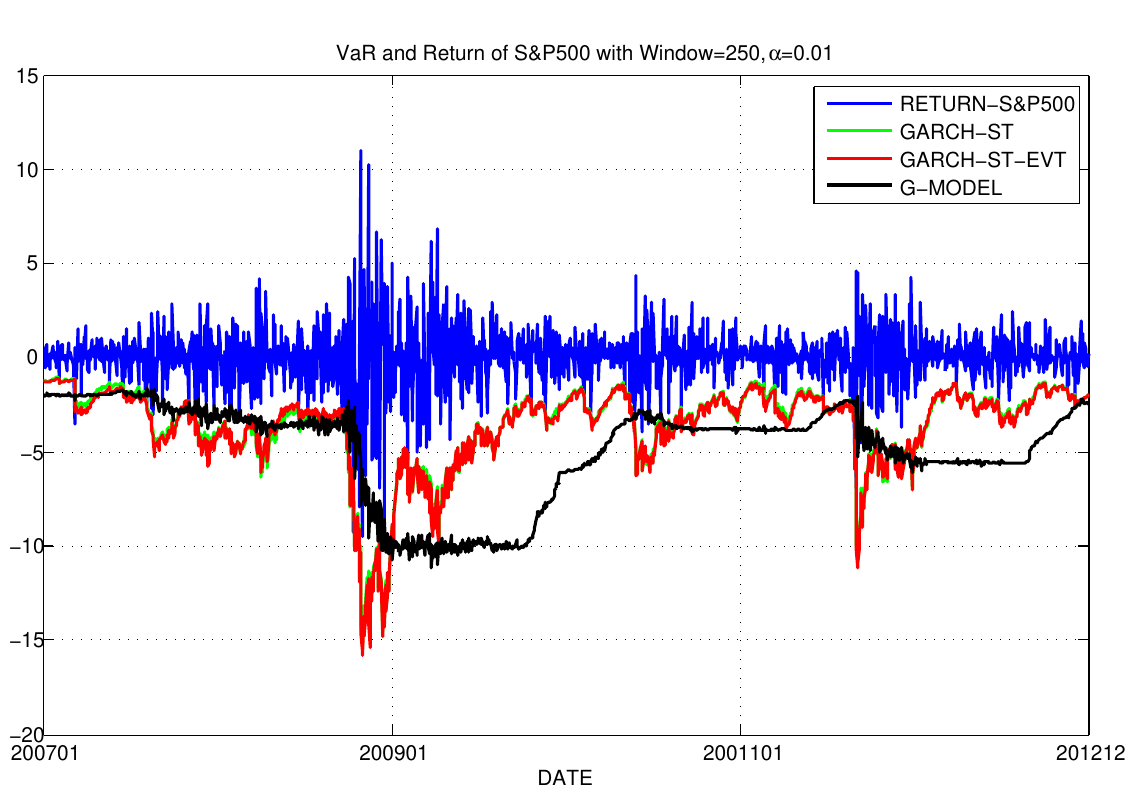}\\
  \caption{For a given $W=1000,500,250$, the VaR performance  of
    different models at risk level $\alpha=0.01$..  \label{sw1}}
\end{figure}

\newpage\thispagestyle{empty}
\begin{figure}[H]
  \centering
\includegraphics[width=3.8 in]{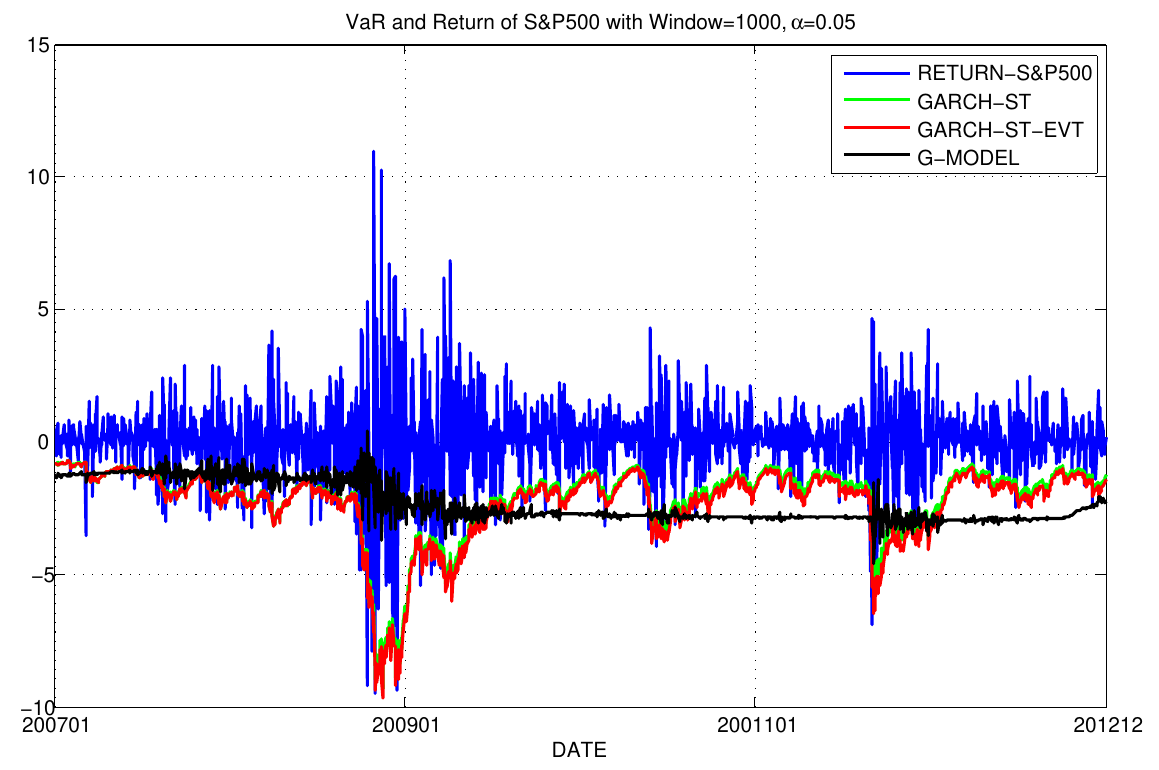}\\
\includegraphics[width=3.8 in]{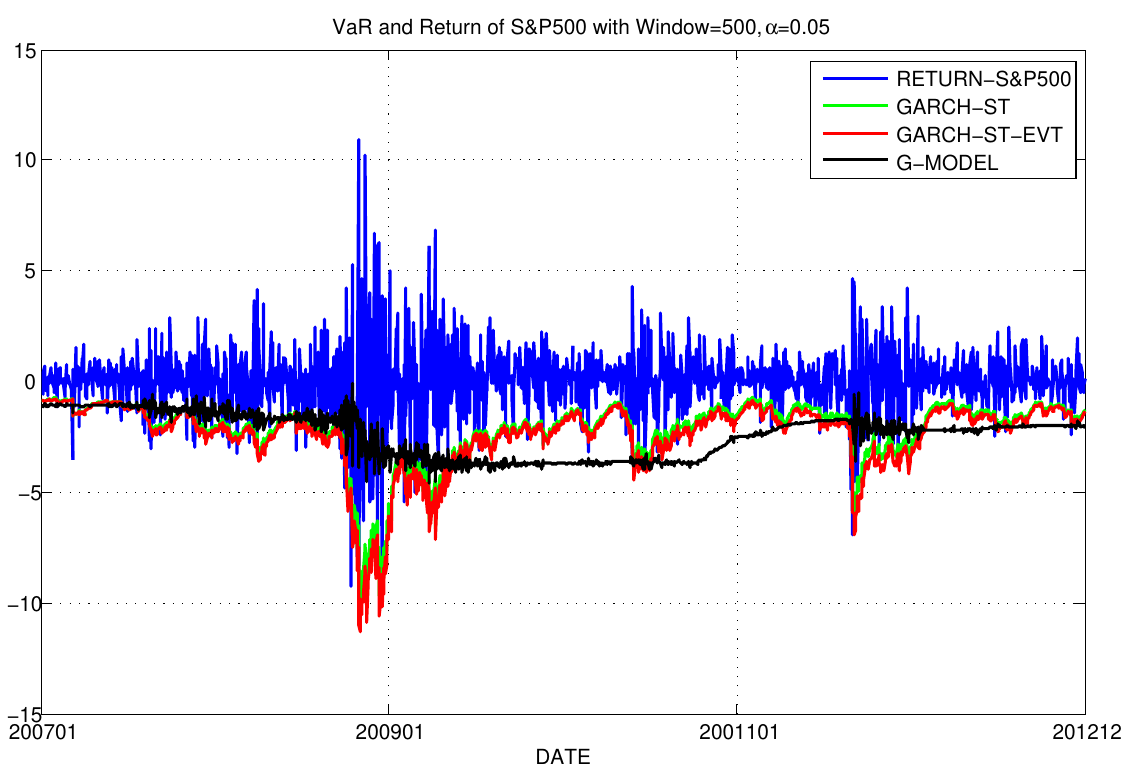}\\
\includegraphics[width=3.8 in]{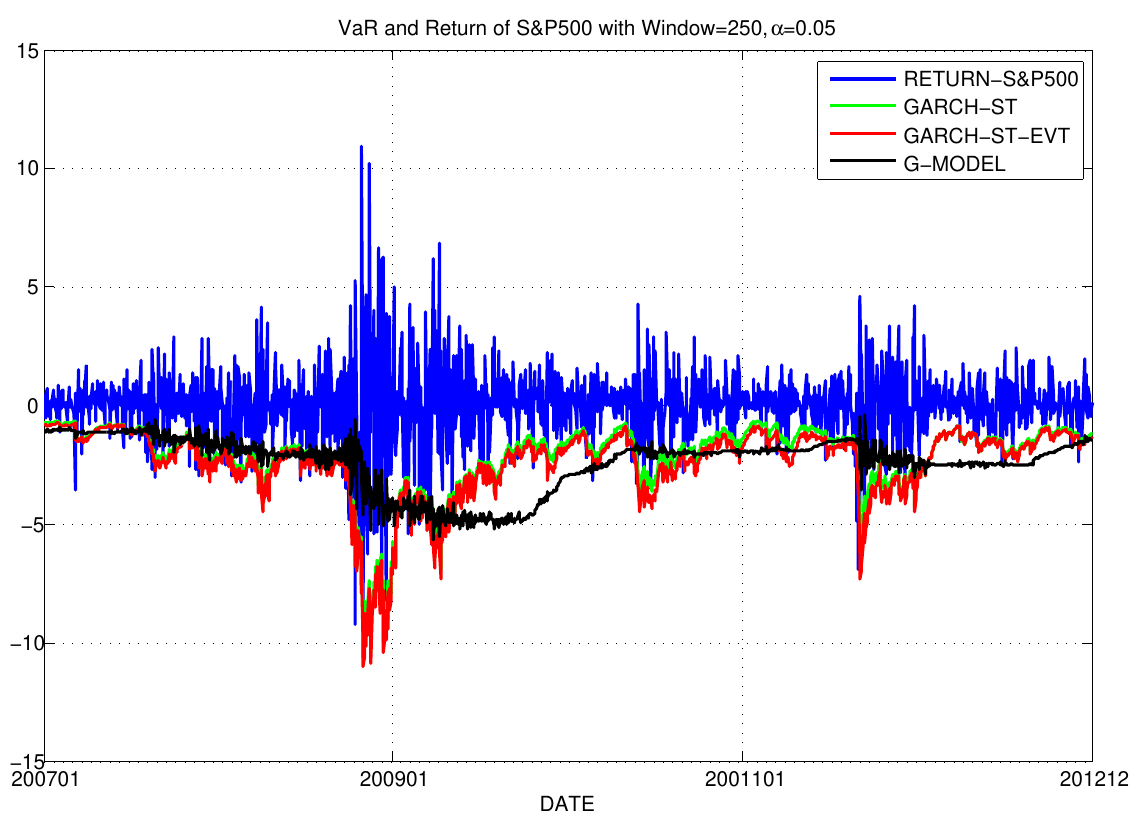}
  \caption{For a given $W=1000,500,250$, the VaR performance of
    different models  at risk level $\alpha=0.05$.  \label{sw5}}
\end{figure}

\newpage
\begin{figure}[H]
  \centering
\includegraphics[width=3.6 in]{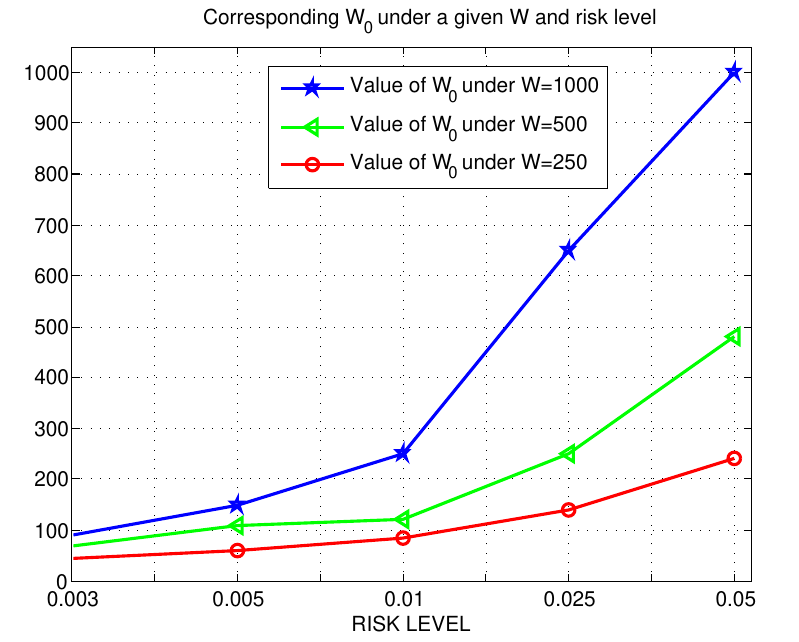}
  \caption{S\&P 500 Index data: variation of  adaptive window
    $W_0$ for different risk levels $\alpha$ and historical windows $W$.  \label{w0}}
\end{figure}

\begin{figure}[H]
  \centering
\includegraphics[width=3.7 in]{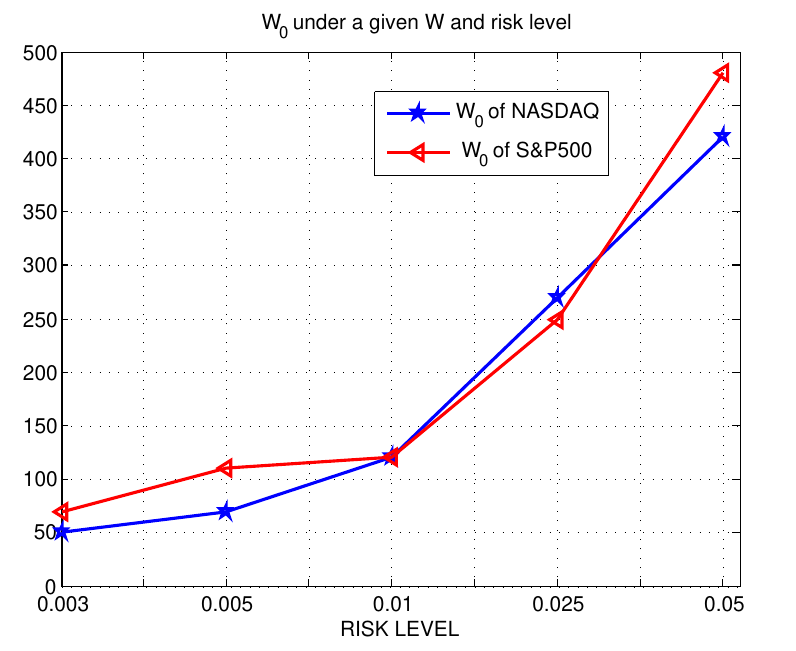}
  \caption{For a given $W=500,\alpha=0.003,0.005,0.01,0.025,0.05$, $W_0$ of NASDAQ  and S\&P500 indexes.  \label{np500}}
\end{figure}


\newpage

\begin{table}[H]
  \centering
  \caption{NASDAQ Composite Index:  Average  and standard deviations of
    \%Viol with  $W$=1000, 500, 250 and  $\alpha=0.01$.}
  \label{table:ns}
  \begin{tabular}{lccccc}
    \toprule
     Model &Window $W$ & Window $W_0$ & Average value & Standard deviation\\
    \midrule
    G-VaR:&
            1000
                          &  350
                                         &  0.0087
                                                         &8.1160e-04\\
              &500
                          &  120
                                         & 0.0103
                                                         & 3.9274e-04\\
              &250
                          &  75
                                         &  0.0104
                                                         &  5.9560e-04\\
    \bottomrule
    \hline
  \end{tabular}
\end{table}

\newpage
\begin{table}[H]
  \centering
  \caption{\small NASDAQ Composite Index: Empirical statistics of  G-VaR
    forecast compared with forecasts of  three benchmark predictors
    reported in
    \cite{Kuester06} with   $W$=1000. The bottom plot shows the $\text{LR}_{uc}$, where $\text{LR}_{uc}$ is the p-values of the Binomial test.}
  \label{table:ns1000}
  \begin{tabular}{lcccc}
    \toprule
    Model & 100$\alpha$ & \%Viol. & $\text{LR}_{uc}$& $100\overline{\text{VaR}}$\\
    \midrule
    AR(1)-GARCH(1,1)-N: $\left\{\tabincell{c}{  \\ \\ \\  } \right.$
          &  \tabincell{c}{ 1 \\ 2.5 \\ 5}
          & \tabincell{c}{ 2.23\\3.92 \\  6.21}
          &\tabincell{c}{0.00\\ 0.00 \\ 0.00}
          & \tabincell{c}{2.05\\1.72 \\ 1.43} \\
    AR(1)-GARCH(1,1)-St: $\left\{\tabincell{c}{  \\  \\ \\} \right.$
          & \tabincell{c}{ 1\\ 2.5\\  5}
          & \tabincell{c}{ 1.2\\ 2.72\\ 5.12}
          &\tabincell{c}{ 0.12\\ 0.25 \\ 0.65}
          & \tabincell{c}{ 2.57\\2.01  \\ 1.59} \\
    AR(1)-GARCH(1,1)-St-EVT: $\left\{\tabincell{c}{  \\ \\ \\  } \right.$
          & \tabincell{c}{ 1\\ 2.5\\  5}
          & \tabincell{c}{ 0.97\\2.47 \\ 5.06}
          &\tabincell{c}{ 0.82\\ 0.87\\0.82}
          & \tabincell{c}{ 2.70\\ 2.07\\ 1.61} \\
    G-VaR:\quad $W_0$=$\left\{ \tabincell{c}{ 350\\650 \\ 900} \right.$
          &  \tabincell{c}{ 1\\ 2.5\\ 5}
          & \tabincell{c}{ 0.99\\ 2.51\\ 5.03}
          &\tabincell{c}{ 0.93 \\ 0.96\\ 0.90 }
          & \tabincell{c}{2.78\\ 2.06\\ 1.52 }\\
    \bottomrule
    \hline
  \end{tabular}
  \vskip5mm
  \includegraphics[width=4.0 in]{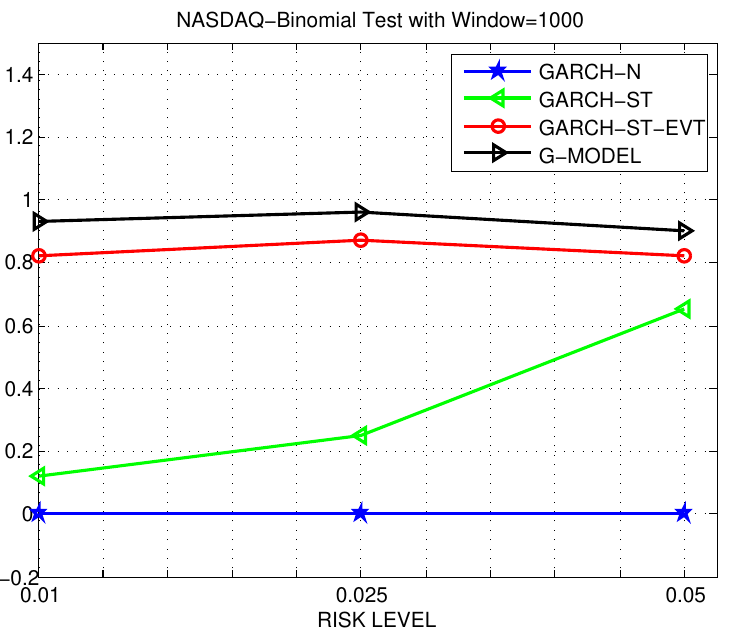}%
\end{table}

\newpage
\begin{table}[H]
\centering
\caption{NASDAQ Composite Index with W=500, 250}
\label{table:ns500250}
\begin{tabular}{lccccc}
\toprule
Model & 100$\alpha$ & \%Viol. & $\text{LR}_{uc}$& $100\overline{\text{VaR}}$\\
\midrule
$W$=500:\quad $W_0$=$\left\{ \tabincell{c}{50\\ 70 \\ 120\\270 \\ 420} \right.$
        &  \tabincell{c}{0.3\\ 0.5 \\ 1\\ 2.5\\ 5}
         & \tabincell{c}{0.30\\ 0.49\\ 1.05\\ 2.54\\ 5.00}
         &\tabincell{c}{0.96\\ 0.95\\ 0.70 \\ 0.81\\ 0.99 }
         & \tabincell{c}{4.71\\ 4.00 \\ 3.11\\ 2.14\\ 1.60 }\\
$W$=250:\quad $W_0$=$\left\{ \tabincell{c}{35\\ 50 \\ 75\\150 \\ 210} \right.$
        &  \tabincell{c}{0.3\\ 0.5 \\ 1\\ 2.5\\ 5}
         & \tabincell{c}{0.30\\ 0.52\\ 1.02\\ 2.49\\ 5.05}
         &\tabincell{c}{1.00\\ 0.82\\ 0.84 \\ 0.93\\ 0.84 }
         & \tabincell{c}{4.29\\ 3.67 \\ 2.98\\ 2.12\\ 1.61 }\\
\bottomrule
\hline
\end{tabular}
\end{table}

\vskip2cm
\begin{table}[H]
\centering
\caption{S\&P500 Index:  Average  and standard deviations of \%Viol
  with
  $W$=1000, 500, 250 and  $\alpha=0.01$}
\label{table:sp}
\begin{tabular}{lccccc}
\toprule
   Model &Window $W$ & Window $W_0$ & Average value & Standard deviation\\
  \midrule
  G-VaR:& 1000        &  250         &  0.0111         &4.1705e-04\\
            &500        &  120         & 0.0097         & 3.6689e-04\\
            &250        &  85         &  0.0099         &     3.0967e-04\\
  \bottomrule
  \hline
\end{tabular}
\end{table}

\newpage\thispagestyle{empty}
\begin{table}[H]
\centering
\caption{\small S\&P 500 Index: Empirical statistics of VaR
  forecasts from G-VaR and three benchmark predictors with  $W$=1000. The bottom plot shows the $\text{LR}_{uc}$.}
\label{table:sp1000}
\begin{tabular}{lccccc}
\toprule
Model & 100$\alpha$ & \%Viol. & $\text{LR}_{uc}$& $100\overline{\text{VaR}}$\\
\midrule

AR(1)-GARCH(1,1)-N: $\left\{\tabincell{c}{ \\  \\  \\ \\ \\  } \right.$
         &  \tabincell{c}{0.3\\ 0.5 \\ 1 \\ 2.5 \\ 5}
         & \tabincell{c}{1.15\\ 1.55 \\ 2.42\\3.83 \\  6.08}
         &\tabincell{c}{0.00\\ 0.00 \\ 0.00\\ 0.00 \\ 0.00}
         & \tabincell{c}{2.64\\ 2.47 \\ 2.23\\1.87 \\ 1.56} \\
AR(1)-GARCH(1,1)-St: $\left\{\tabincell{c}{ \\  \\  \\  \\ \\} \right.$
        & \tabincell{c}{0.3\\ 0.5 \\ 1\\ 2.5\\  5}
        & \tabincell{c}{0.28\\ 0.73 \\ 1.32\\ 3.24\\ 5.71}
        &\tabincell{c}{0.83\\ 0.07\\ 0.07\\ 0.01 \\ 0.06}
        & \tabincell{c}{3.42\\ 3.06 \\ 2.61\\2.03  \\ 1.60} \\
AR(1)-GARCH(1,1)-St-EVT: $\left\{\tabincell{c}{ \\  \\  \\ \\ \\  } \right.$
        & \tabincell{c}{0.3\\ 0.5 \\ 1\\ 2.5\\  5}
        & \tabincell{c}{0.39\\ 0.62 \\ 1.21\\ 2.79\\ 4.73}
        &\tabincell{c}{0.33\\ 0.33\\ 0.33\\0.28 \\  0.45}
        & \tabincell{c}{3.38\\ 3.10 \\ 2.70\\2.16 \\1.72} \\
G-VaR:\quad $W_0$=$\left\{ \tabincell{c}{90\\ 150 \\ 250\\650 \\ 1000} \right.$
        &  \tabincell{c}{0.3\\ 0.5 \\ 1\\ 2.5\\ 5}
         & \tabincell{c}{0.29\\ 0.52\\ 1.07\\ 2.49\\ 4.87}
         &\tabincell{c}{0.91\\ 0.86\\ 0.68 \\ 0.97\\ 0.72 }
         & \tabincell{c}{7.05\\ 5.77 \\ 4.40\\ 2.91\\ 1.94 }\\
\bottomrule
\hline
\end{tabular}\\
\includegraphics[width=3.8 in]{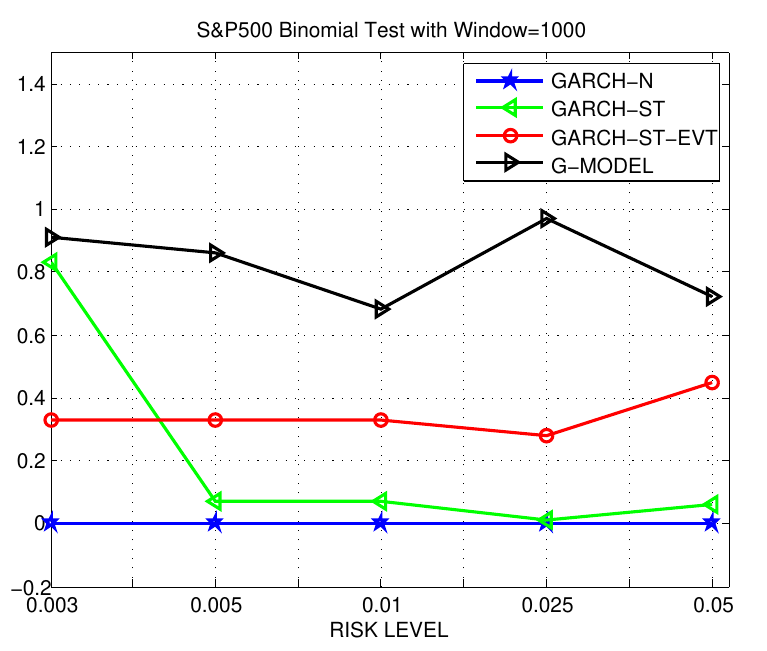}
\end{table}

\newpage\thispagestyle{empty}
\begin{table}[H]
\centering
\caption{\small S\&P 500 Index: Empirical statistics of VaR
  forecasts from G-VaR and three benchmark predictors with   $W$=500. The bottom plot shows the $\text{LR}_{uc}$.}
\label{table:sp500}
\begin{tabular}{lccccc}
\toprule
Model & 100$\alpha$ & \%Viol. & $\text{LR}_{uc}$& $100\overline{\text{VaR}}$\\
\midrule

AR(1)-GARCH(1,1)-N: $\left\{\tabincell{c}{ \\  \\ \\  \\ \\  } \right.$
         &  \tabincell{c}{0.3\\ 0.5 \\ 1 \\ 2.5 \\ 5}
         & \tabincell{c}{1.06\\ 1.38 \\ 2.22\\ 3.82\\  5.90}
         &\tabincell{c}{0.00\\ 0.00 \\ 0.00\\ 0.00 \\ 0.01}
         & \tabincell{c}{2.75\\ 2.58\\ 2.32\\ 1.95\\ 1.63} \\
AR(1)-GARCH(1,1)-St: $\left\{\tabincell{c}{ \\  \\ \\ \\   \\} \right.$
        & \tabincell{c}{0.3\\ 0.5 \\ 1\\ 2.5\\  5}
        & \tabincell{c}{0.27\\ 0.59 \\ 1.18\\ 3.13\\ 5.67}
        &\tabincell{c}{0.74\\ 0.42\\ 0.25\\ 0.01 \\ 0.05}
        & \tabincell{c}{3.51\\ 3.16 \\ 2.69\\2.11  \\ 1.67} \\
AR(1)-GARCH(1,1)-St-EVT: $\left\{\tabincell{c}{ \\  \\ \\  \\ \\  } \right.$
        & \tabincell{c}{0.3\\ 0.5 \\ 1\\ 2.5\\  5}
        & \tabincell{c}{0.37\\ 0.62 \\ 1.06\\ 2.57\\ 5.01}
        &\tabincell{c}{0.43\\ 0.31\\ 0.70\\ 0.79 \\ 0.98}
        & \tabincell{c}{3.44\\ 3.16 \\ 2.77\\ 2.23  \\ 1.80} \\
G-VaR:\quad $W_0$=$\left\{ \tabincell{c}{70\\ 110 \\ 120\\ 250 \\ 480} \right.$
        &  \tabincell{c}{0.3\\ 0.5 \\ 1\\ 2.5\\ 5}
         & \tabincell{c}{0.33\\ 0.51 \\ 0.96\\2.48 \\ 5.08}
         &\tabincell{c}{0.74\\ 0.96\\ 0.81 \\ 0.90\\ 0.81 }
         & \tabincell{c}{5.50\\ 4.58 \\ 4.08\\ 2.79\\ 1.90 }\\
\bottomrule
\hline
\end{tabular}
\vskip5mm
\includegraphics[width=3.8 in]{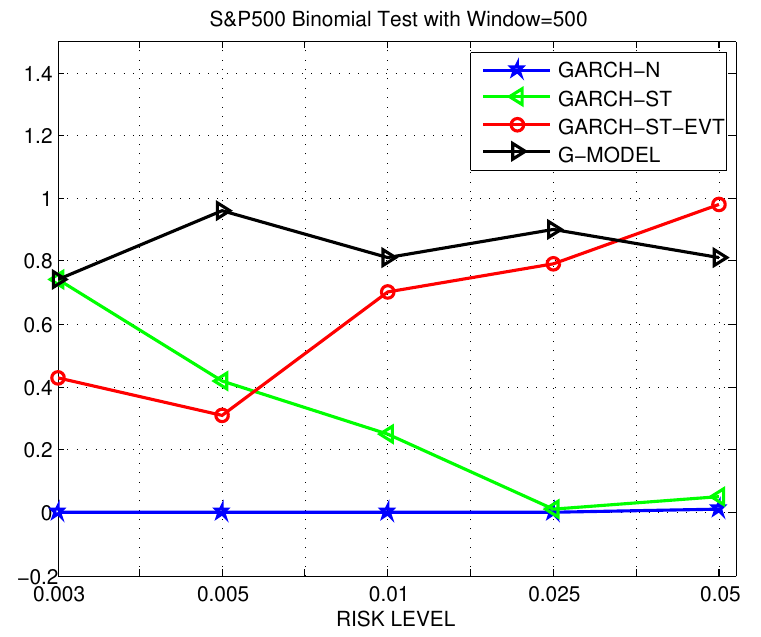}
\end{table}

\newpage\thispagestyle{empty}
\begin{table}[H]
\centering
\caption{\small S\&P 500 Index: Empirical statistics of VaR
  forecasts from G-VaR and three benchmark predictors with   $W$=250. The bottom plot shows the $\text{LR}_{uc}$.}
\label{table:sp250}
\begin{tabular}{lccccc}
\toprule
Model & 100$\alpha$ & \%Viol. & $\text{LR}_{uc}$& $100\overline{\text{VaR}}$\\
\midrule

AR(1)-GARCH(1,1)-N: $\left\{\tabincell{c}{ \\  \\  \\ \\ \\  } \right.$
         &  \tabincell{c}{0.3\\ 0.5 \\ 1 \\ 2.5 \\ 5}
         & \tabincell{c}{1.30\\ 1.72 \\ 2.56\\ 4.09\\  6.23}
         &\tabincell{c}{0.00\\ 0.00 \\ 0.00\\ 0.00 \\ 0.00}
         & \tabincell{c}{2.78\\ 2.61 \\ 2.35\\ 1.97\\ 1.65} \\
AR(1)-GARCH(1,1)-St: $\left\{\tabincell{c}{ \\    \\ \\ \\ \\} \right.$
        & \tabincell{c}{0.3\\ 0.5 \\ 1\\ 2.5\\  5}
        & \tabincell{c}{0.40\\ 0.70 \\ 1.39\\ 3.3\\ 5.95}
        &\tabincell{c}{0.28\\ 0.08\\ 0.01\\ 0.00 \\ 0.01}
        & \tabincell{c}{3.47\\ 3.13 \\ 2.68\\ 2.12 \\ 1.69} \\
AR(1)-GARCH(1,1)-St-EVT: $\left\{\tabincell{c}{ \\  \\ \\ \\ \\  } \right.$
        & \tabincell{c}{0.3\\ 0.5 \\ 1\\ 2.5\\  5}
        & \tabincell{c}{0.65\\ 0.86 \\ 1.40\\2.90 \\ 5.18}
        &\tabincell{c}{0.00\\ 0.00\\ 0.01\\0.10 \\ 0.59}
        & \tabincell{c}{3.32\\ 3.06 \\ 2.71\\2.22 \\ 1.80} \\
G-VaR:\quad $W_0$=$\left\{ \tabincell{c}{45\\ 60 \\ 85\\ 140 \\ 240} \right.$
        &  \tabincell{c}{0.3\\ 0.5 \\ 1\\ 2.5\\ 5}
         & \tabincell{c}{0.29\\ 0.48 \\ 0.98\\ 2.55\\ 4.95}
         &\tabincell{c}{0.86\\ 0.82\\ 0.87 \\ 0.85\\ 0.88 }
         & \tabincell{c}{4.73\\ 4.16 \\ 3.46\\ 2.57\\ 1.83 }\\
\bottomrule
\hline
\end{tabular}
\vskip 5mm
\includegraphics[width=3.8 in]{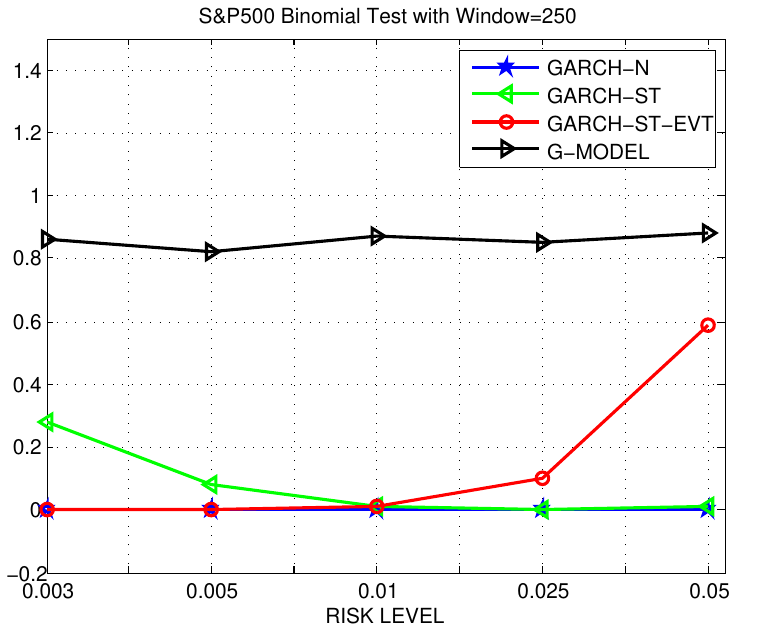}
\end{table}

\end{document}